%% file: main.tex
\documentclass[11pt]{article}
\usepackage{amssymb,amsmath,amsfonts,amssymb,amsthm}
\usepackage{graphics,graphicx,color}
\usepackage{empheq}
\usepackage{slashed}
\usepackage{tikz,tkz-fct}

\usepackage{hyperref} 
\usepackage{algorithm}
\usepackage{algpseudocode}
\usepackage{footnote}
\usepackage{caption,subcaption}

\usepackage[margin=1.0in]{geometry}

\newtheorem{theorem}{Theorem}[section]
\newtheorem{lemma}[theorem]{Lemma}
\newtheorem{proposition}[theorem]{Proposition}

\newtheorem{remark}[theorem]{Remark}
\newtheorem{hypothesis}[theorem]{Hypothesis}

\input{self_defined_command_gb}
\newcommand{\sm}{{s_m}}

\title{$\Zm_2$ classification of FTR symmetric differential operators and obstruction to Anderson localization}
\author{Guillaume Bal \thanks{Departments of Statistics and Mathematics and Committee on Computational and Applied Mathematics, University of Chicago, Chicago, IL 60637; guillaumebal@uchicago.edu} \and Zhongjian Wang \thanks{Division of Mathematical Sciences, School of Physical and Mathematical Sciences, Nanyang Technological University, Singapore 637371; zhongjian.wang@ntu.edu.sg} }

\begin{document}

\maketitle

\begin{abstract}
 This paper concerns the $\Zm_2$ classification of Fermionic Time-Reversal (FTR) symmetric partial differential Hamiltonians on the Euclidean plane. We consider the setting of two insulators separated by an interface. Hamiltonians that are invariant with respect to spatial translations along the interface are classified into two categories depending on whether they may or may not be gapped by continuous deformations. Introducing a related odd-symmetric Fredholm operator, we show that the classification is stable against FTR-symmetric perturbations. 
 
 The property that non-trivial Hamiltonians cannot be gapped may be interpreted as a topological obstruction to Anderson localization: no matter how much (spatially compactly supported) perturbations are present in the system, a certain amount of transmission in both directions is guaranteed in the nontrivial phase. We present a scattering theory for such systems and show numerically that transmission is indeed guaranteed in the presence of FTR-symmetric perturbations while it no longer is for non-symmetric fluctuations.
\end{abstract}


\section{Introduction}
\label{sec:introduction}
A salient feature of (Chern) topological insulators is the robust asymmetric transport observed at one-dimensional interfaces separating distinct two-dimensional insulators. In many settings, this robustness may be given a topological interpretation. For large classes of discrete and continuous Hamiltonians, a quantized interface conductivity as well as a number of Fredholm operators with non-trivial indices may be introduced to compute topological invariants and relate them to the observed quantized, robust-to-perturbations, asymmetric transport along the interface; see, e.g. \cite{Elbau,PS,SB-2000} for discrete Hamiltonians and \cite{2,bal2022topological,bal2023topological,Drouot:19,Drouot:19b,drouot2020microlocal,Hatsugai,QB22} for partial- or pseudo- differential Hamiltonians. More general results on the vast and active field of topological phases of matter may be found in, e.g.,  \cite{ASS90,BH,bourne2017k,chiu2016classification,delplace,kitaev2009periodic,moessner2021topological,thiang2016k,Volovik,WI}.

The above robust asymmetric transport is surprising as it may be seen as an obstruction to the Anderson localization observed in topologically trivial systems; see \cite{aizenman1993localization,carmona2012spectral,fouque2007wave,frohlich1983absence,germinet2001bootstrap,lagendijk2009fifty,sheng2007introduction} for pointers to the large bibliography on Anderson localization and \cite{topological,ludwig2013lyapunov,PS} for references on such topological obstructions. A possible interpretation of Anderson localization is that transmission through a slab of random heterogeneities is exponentially small in the thickness of the slab with a coefficient of proportionality related to the localization length; see, e.g., \cite{fouque2007wave}. A non-trivial interface conductivity implies that robust transport in at least one direction occurs independently of the thickness of the random slab, which leads to a quantized obstruction to Anderson localization; see \cite{topological}  where the phenomenon is analyzed for randomly perturbed Dirac models.

For time-reversal symmetric Hamiltonians, this interface transport is no longer asymmetric. Yet, robust transport may still be observed even in the presence of strong random fluctuations. Since the above invariants based on indices of Fredholm operators necessarily vanish, they cannot explain such transport immunity to random fluctuations. An other index, based on symmetry protection and taking for the form of a $\Zm_2$ invariant, may then be introduced. A $\Zm_2$ index classifies systems into two classes: trivial or non-trivial. Such a symmetry protection occurs for Hamiltonians that are Fermionic Time Reversal (FTR) symmetric. This symmetry is characterized by an anti-unitary operator $\theta$ such that $\theta^2=-I$. FTR-symmetric Hamiltonians $H$ are then those obeying the commutation relation: $\theta H = H\theta$.

Following the pioneering works in \cite{PhysRevLett.95.146802}, crystals with FTR-symmetric Hamiltonians have been analyzed in great detail in many works; see, e.g., \cite{BH,PhysRevB.76.045302,RevModPhys.82.3045,schulz2015z2}, with techniques allowing us to compute the class of a bulk Hamiltonian from its spectral decomposition on a Brillouin zone. The FTR-symmetry is one among a large class of possible symmetries \cite{chiu2016classification,kitaev2009periodic}. It is natural in electronic structures with strong spin-orbit couplings \cite{BH}. This symmetry should not be confused with the more standard Time Reversal symmetry, characterized by an anti-unitary operator $\tau$ such that $\tau^2=1$, and for which the results of this paper do not apply.

\medskip

The main objective of this paper is to propose a $\Zm_2$ classification for partial differential Hamiltonians modeling a transition between two insulators. We do not identify bulk phases but rather assign a topological invariant to the interface separating them. This follows the framework proposed in \cite{bal2022topological,bal2023topological,QB22}, where a $\Zm$ classification of interface Hamiltonians is obtained even when bulk phases may not be identified. Our approach partially builds on the $\Zm_2$ index for odd-symmetric Fredholm operators in \cite{schulz2015z2}. We consider differential operators acting on (vector-valued) functions of the Euclidean plane parametrized by $(x,y)$. The Hamiltonian models a transition from an insulator in the domain $y\gtrsim1$ to an other insulator in the domain $y\lesssim -1$. Excitations are allowed to propagate along the $x-$axis in the vicinity of $y\approx0$.

We first consider Hamiltonians that are invariant against spatial translations along the $x-$axis  (i.e., commute with spatial translations in $x$) and assume that they may be written in the form $H= H_1\oplus \theta^* H_1\theta$. The operator $H_1$ is typically straightforward to obtain in practical applications. The $\Zm_2$ classification is then based on whether $H$ may be continuously deformed (in the uniform sense) to a gapped system or not along a curve of FRT-symmetric Hamiltonians. A gappable system (for a fixed energy range) is considered trivial while a non-gappable system is considered non-trivial; see also \cite{kane2013topological,schulz2015z2}.

Moreover, the $\Zm_2$ invariant may be computed explicitly from the topological index associated to the (non-FTR-symmetric) operator $H_1$. More precisely, the $\Zm_2$ index is given by $(-1)^{2\pi \sigma_I[H_1]}$ where $\sigma_I[H_1]$ is a bulk-difference (Chern-type) invariant that only depends on the symbol of the operator $H_1$ and may be easily computed, for instance as an application of an explicit Fedosov-H\"ormander formula as developed in  \cite{bal2022topological,bal2023topological,QB22}.

\medskip

Following \cite{schulz2015z2}, we next associate to $H$ a Fredholm operator $T$ and show that the $\Zm_2$ index of $H$ is related to the $\Zm_2$ index of the odd-symmetric Fredholm operator $T$ given by $\ind_2T=$ dim Ker $T$ mod $2$. The non-trivial category is characterized by $\ind_2T=1$ then. We then show that the index of the operator $T$ remains stable when $H$ is deformed to a perturbed operator $H+V$ for $V$ a FTR-symmetric perturbation that is compactly supported in the $x$-variable. This shows the stability of the  $\Zm_2$ index against perturbations.

\medskip

While interesting in its own right, the above $\Zm_2$ index is not directly related to any transport along the $x-$axis or transmission across a slab of random fluctuations. Following the construction in \cite{chen2023scattering}, we propose a scattering theory for a class of Hamiltonians $H$. We present a construction of generalized eigenfunctions in the presence of perturbations, describe how this is used to construct a scattering matrix for such operators, and show how the  $\Zm-$index $2\pi \sigma_I[H_1]$ and the $\Zm_2-$index of $H$ may be computed explicitly from the scattering matrix. 

We finally present a numerical algorithm that allows for a fast and accurate computation of the scattering coefficients. The algorithm is applied to a number of Hamiltonians generalizing the Dirac operators analyzed in \cite{bal2023asymmetric}, to which we refer for some details of the construction. The numerical simulations illustrate the main theoretical findings of this paper, namely that non-trivial Hamiltonians, even in the presence of perturbations, display robust transmission across a slab of randomness. As such, a non-trivial $\Zm_2$ index may be seen as a topological obstruction to Anderson localization, as analyzed for FTR-symmetric versions of the Dirac operator in \cite{topological}. In contrast, FTR-symmetric operators perturbed by non-FTR-symmetric fluctuations no longer display such a robust transmission. For such operator, our numerical simulations display results in perfect agreement with standard Anderson localization (exponential attenuation of the transmitted signal as the thickness of the random slab increases).

\medskip

The rest of the paper is structured as follows. The $\Zm_2$ classification of operators of the form $H=H_1+\theta H_1\theta^*$ and $H_V=H+V$ is presented in section \ref{sec:Z2}. The scattering theory for a subset of such Hamiltonians is then given in section \ref{sec:scattering}.  The scattering matrix directly encodes transmission information across a slab of random perturbations. We present a numerical algorithm in section \ref{sec:num} and numerical computations of scattering matrices in section \ref{sec:examples} as illustrations of the theories on robust transport protected by the FTR symmetry.

\section{$\Zm_2$ invariant and spectral decomposition}
\label{sec:Z2}

\paragraph{Fermionic Time Reversal Symmetry.}

Let $\theta=\mK \rJ$ be an operator modeling Fermionic time reversal symmetry on a complex separable Hilbert space $\mH$ with $\mK$ complex conjugation and $\rJ$ a unitary operator such that $\rJ^2=-\rI$, with $\rI$ the identity matrix, or equivalently $\rJ^*=-\rJ$. Since the eigenvalues of $\rJ$ are $\pm i$, we can always choose a basis such that $\rJ$ takes the form $\rJ=i\sigma_2\otimes \rI=\sigma_+\otimes \rI-\sigma_-\otimes \rI$, where we have defined $\sigma_\pm=\frac12(\sigma_1\pm i\sigma_2)$. Note that $\rJ\mK=\mK \rJ$.  The decomposition implies a direct sum (grading)  $\mH=\mH_+\oplus \mH_-$ such that for $\psi=(\psi_+,\psi_-)^t$ in that direct sum, then $\rJ\psi = (\psi_-,-\psi_+)^t$. 
We verify that $\theta$ is anti-linear, i.e., $\theta(\alpha\psi)=\bar\alpha\theta(\psi)$ and the following identities hold: $\theta^2=-\rI$ and $\theta^*=-\mK J=-J \mK=-\theta$.

An operator $T$ is said to be {\em odd-symmetric} when the following equivalent relations hold:
\begin{equation}\label{eq:oddsymmetric}
  \theta^* T \theta = T^*,\quad T\theta = \theta T^*,\quad -\rJ\bar T \rJ= T^*,\quad -\rJ T \rJ = T^t,\quad T\rJ=\rJ T^t.
\end{equation}
We defined here $\bar T:=\mK T \mK$. 

A Hamiltonian $H$ is said to be {\em Fermionic time-reversal (FTR) symmetric} when it is odd-symmetric and Hermitian. In other words, $H=H^*$ or $\bar H=H^t$ with 
\begin{equation}\label{eq:FTR}
  \theta^* H \theta = H,\quad H\theta = \theta H,\quad -\rJ\bar H \rJ= H,\quad -\rJ H \rJ = \bar H,\quad H\rJ=\rJ\bar H.
\end{equation}

The Hamiltonians of interest here are elliptic differential or pseudo-differential operators, which are self-adjoint unbounded operator on $\mH$ and describe a transition between two insulating regions in two space dimensions. With a parametrization $(x,y)$ of Euclidean space $\Rm^2$, the interface is modeled by $y\sim0$ while the north and south insulating regions are modeled by $\pm y \gtrsim 1$, respectively. Our main objective is to describe current along the $x-$axis for a large class of FTR-symmetric Hamiltonians.

\paragraph{Interface conductivity.}  The two regions away from the interface are assumed insulating in an energy range $I:=[E_-,E_+]$ with $E_-<E_+$. An interface conductivity describing current along this interface is defined as \cite{Elbau,PS,SB-2000}
\begin{equation}\label{eq:sigmaI}
 \sigma_I[H] = \Tr\ i[H,P] \varphi'(H).
\end{equation}
The observable $i[H,P]$ is a current operator (from left to right along the $x$ axis) with $P=P(x)\in \mathfrak{S}[0,1]$ while $\varphi'(H)$ is a density supported in the energy range $I$ and integrating to $1$, in other words $\varphi\in \mathfrak{S}[0,1;E_-,E_+]$. Here, we denote by $\mathfrak{S}[c_1,c_2;\lambda_1,\lambda_2]$ the set of {\em switch} functions, i.e., real-valued functions $f$ on $\mathbb{R}$ for which there exists $\delta>0$ such that $f(x)=c_1$ for all $x\leq\lambda_1+\delta$ and $f(x)=c_2$ for all $x\geq\lambda_2-\delta$. The union over $\lambda_1<\lambda_2$ is denoted by $\mathfrak{S}[c_1,c_2]$.  We consider Hamiltonians such that $i[H,P] \varphi'(H)$ is a trace-class operator (compact operator with summable singular values) and $2\pi\sigma_I\in\Zm$ is defined and quantized as shown in \cite{bal2022topological,bal2023topological,QB22}. 

However, for FTR-symmetric operators, the above conductivity vanishes. 
\begin{lemma} \label{lem:symsigma} 
For self-adjoint Hamiltonians on $\mH$ such that $\theta^* H \theta = H$ and $i[H,P] \varphi'(H)$ is a trace-class operator, we have $\sigma_I[H]=0$.
\end{lemma}
\begin{proof} We observe by an application of the spectral theorem that $\varphi'(\theta^* H \theta)=\theta^* \varphi'(H)\theta$ since $\varphi'(h)$ is real-valued. Also, $[P,\theta]=0$ since $P$ is a scalar real-valued function so that $[\theta^*H\theta,P]=\theta^*[H,P]\theta$. We obtain that
\[
  \sigma_I[H] = \sigma_I[\theta^*H\theta]= \Tr\ i \theta^*[H,P]\theta \theta^*\varphi'(H)\theta = -  \Tr\ \theta^* i[H,P]\varphi'(H)\theta  = -\sigma_I[H].
\]
The last equality comes from the fact that for $A$ a trace-class operator,
\[
  \Tr\ \theta^* A \theta  = \sum_i (\theta e_i,A \theta e_i) = \sum_i (e_i,A e_i) =  \Tr \ A
\]
since $\{e_i\}$ is an o.n.b. when $\{\theta e_i\}$ is. Here, $(\cdot,\cdot)$ is the inner product on $\mH$.
This completes the proof of the result.
\end{proof}

\paragraph{Spectral decomposition.}   

We first consider operators $H$ defined as unbounded differential (or pseudodiffential) operators on $L^2(\Rm^2;\Cm^n)$ that commute with translations in the first variable $x$, with $(x,y)\in \Rm^2$ a parametrization of the Euclidean plane and $n\geq1$ an integer. We denote by $\xi$ the dual variable to $x$.  

We further assume that $H$ has spectrum solely composed of absolutely continuous spectrum in the interval $I=[E_-,E_+]$. We denote by $H_{I}= H\Pi_{I}[H]$ with $\Pi_{I}[H]$ orthogonal projection onto the space generated by the (absolutely continuous) spectrum of $H$ in $I$.  Since $H$ is insulating away from $y\sim0$, the a.c. spectrum describes transport along the $x-$axis. Using the multiplicative version of the spectral theorem, this operator may be decomposed as 
\begin{equation}\label{eq:spectral}
  H_I =\dsum_j\dint_{\Xi_j}  E_j(\xi) \Pi_j(\xi) d\xi.
\end{equation}
Here, we assume that $j$ runs over a finite set of indices while $\Rm\ni \Xi_j$ are (open) intervals in $\Rm$ and $d\xi$ is Lebesgue measure on them.  We assume in this spectral representation that the spectral density $\Pi_j(\xi)$ is formally a rank-one operator while the sum over $j$ represents the multiplicity of the spectrum at an energy $E_j(\xi)=E\in I$. A decomposition of the form \eqref{eq:spectral} is always available as per the spectral theorem. The main assumption is that the branches of a.c. spectrum are parametrized by a one dimensional parameter $\xi$ and that the degeneracy (the sum over $j$) is discrete for finite intervals $I$. 

Because $\theta^* H \theta= H$, the projector-valued measure $\theta^*\Pi_j(\xi)\theta d\xi$ must be present in the above spectral representation and associated to the same energy $E=E_j(\xi)$. Indeed, formally $H\theta\Pi_j\theta^*=\theta H\Pi_j\theta^*=E_j\theta\Pi_j\theta^*$. Moreover, for any $f\in \mH$, we observe that 
\[
  (f,\theta f)=(-\theta \theta f,\theta f)=-\overline{(\theta f,f)} = -(f,\theta f)=0
\]
so that $\Pi_j(\xi)d\xi$ and $\theta^*\Pi_j(\xi)\theta d\xi$ must project onto orthogonal subspaces of $\mH$. This Kramers degeneracy is at the origin of the $\Zm_2$ invariant. It does not hold for $\theta$ replaced by an anti-unitary operator $\tau$ such that $\tau^2=1$, which is modeling standard time-reversal symmetry instead of $\theta^2=-1$ modeling a FTR-symmetry.

The FTR symmetry therefore intuitively implies that we must have after some relabeling of the branches the following spectral decomposition:
\begin{equation}\label{eq:spectraldecomp}
  H_I =\dsum_j\dint_{\Xi_j}  E_j(\xi)  \Big(\Pi_j(\xi) +\theta^*\Pi_j(\xi)\theta \Big) d\xi =: H_1 + \theta^* H_1\theta.
\end{equation}
Our main {\em assumption} is that we can find an operator $H_1$ (we drop the implicit dependence on the interval $I$) as above with the sum over $j$ is finite and each branch $\xi\to E_j(\xi)$  sufficiently smooth and such that $E_j(\partial \Xi_j)\in \{E_-,E_+\}$.  This implies that the branches $E_j(\xi)$ leave the interval $(E_-,E_+)$ for $\xi$ a boundary point of the interval $\Xi_j$. Such a property holds for instance for $H$ an elliptic operator or when no flat band such as a Landau level is allowed in the interval $[E_-,E_+]$. This will ensure that $[H_1,P]\varphi'(H_1)$ is trace-class and $\sigma_I[H_1]$ is well-defined. In practical applications, it is often straightforward to come up with such a decomposition and show that the branches $\xi\to E_j(\xi)$ are real-analytic.

\begin{hypothesis}[H1] \label{hyp:H1}
 We collect our main assumptions: $H$ is an unbounded (pseudo-) differential operator that commutes with translations in the $x$ variable and whose spectrum restricted to the interval $I$ takes the form $H_1+\theta^* H_1\theta$ as defined in \eqref{eq:spectraldecomp} with $j$ running over a finite set, $\Xi_j$ being an open interval in $\Rm$, $E_j(\xi)$ being real-analytic on $\Xi_j$, and $\Pi_j(\xi)$ describing simple spectrum. Finally, $E_j(\partial \Xi_j)\in \{E_-,E_+\}$.
\end{hypothesis}

By simple spectrum, we mean that $\Pi_j(\xi)$ has a Schwartz kernel of the form $\psi(x,y;\xi)\psi^*(x',y';\xi)$ with $\psi(x,y;\xi)=(2\pi)^{-\frac12}e^{i x\xi} \phi(y;\xi)$ and $\phi(y,\xi)\in L^2(\Rm;\Cm^n)$. This decomposition comes from the invariance by translation of the operator $H$.
Note that the spectral decomposition \eqref{eq:spectraldecomp} implies that for any indices $j$ and $k$ and any $\xi\in \Xi_j$ and $\xi'\in \Xi_k$ with $\xi\not=\xi'$ when $j=k$, then
\[
  \Pi_j(\xi) \theta^*\Pi_k(\xi')\theta  = \theta^*\Pi_k(\xi')\theta\Pi_j(\xi)=0.
\]
In particular, we have a natural decomposition $\mH=\mH_1+\mH_1^\perp$ with $H_1$ acting on $\mH_1$ and $\theta^* H_1\theta$ on $\mH_1^\perp$.  In the special, but practically useful, case where where $\mH_+=\mH_1$, then we observe that $H$ (restricted to its ac spectrum on $I$) admits the direct sum $H_1\oplus \bar H_1$ since
\[
  \theta^* \begin{pmatrix} H_1 & 0 \\ 0 & H_2 \end{pmatrix} \theta = \begin{pmatrix} \bar H_2 & 0 \\ 0 & \bar H_1 \end{pmatrix}.
\]
While we do not require this additional symmetry, it is satisfied by all the practical examples we consider later in the paper.

The operator $H_1$ in \eqref{eq:spectraldecomp} is {\em not} uniquely defined by $H$ even when $\mH_+=\mH_1$. Indeed, consider the operator $H=D\oplus D\oplus \bar D\oplus \bar D$ in one dimension with $D=-i\partial_x$. This is an unbounded operator on $\mH=L^2(\Rm;\Cm^4)$ with domain of definition $H^1(\Rm;\Cm^4)$. Then
\[
  H = (D\oplus D) \oplus \overline{(D\oplus D)} \equiv (D\oplus \bar D) \oplus \overline{(D\oplus \bar D)}
\]
so that we may choose $H_1$ as either $D\oplus D$ or $D\oplus \bar D$. Here $A\equiv B$ means $A= UBU^*$ for some unitary transformation $U$ (above a permutation). Any topological classification of $H$ will have to be independent of the choice of the operator $H_1$ as we will demonstrate later in this section and for instance in Theorem \ref{thm:Z2}. 

\paragraph{Conductivity and spectral flow.} The conductivity $\sigma_I$ of $H$ and $H_1$ may be directly related to the spectral flows of branches of absolutely continuous spectrum in \eqref{eq:spectraldecomp}.  Let us introduce the decomposition
\[
  H_j= H_{1j}+ \theta^* H_{1j}\theta, \quad H_{1j} = \dint_{\Xi_j}  E_j(\xi) \Pi_j(\xi) d\xi. 
\]
For each branch $j$, we let $\Xi_j=(\xi_{j-},\xi_{j+})$ and define the {\bf spectral flow} $SF_j=SF(H_{1j})$ as 
\begin{equation}\label{eq:SF}
  SF(H_{1j})=  \sgn{E_j(\xi_{j+}) -E_j(\xi_{j-}) }.
\end{equation}
In other words, $SF_j=1$ when $E_j(\xi)$ crosses the interval $[E_-,E_+]$ upwards, $SF_j=-1$ when $E_j(\xi)$ crosses the interval $[E_-,E_+]$ downwards, and $SF_j=0$ otherwise. Define similarly $\widetilde{SF}_j= SF(\theta^* H_{1j}\theta)$ for the time-reversal symmetric branches. Then we have the following result
\begin{lemma} \label{lem:SFbranch}
  We have $\sigma_I[H_j] = SF_j$ and $\sigma_I[H_j+\theta^*H_j \theta] = SF_j+\widetilde{SF}_j=0$.
\end{lemma}
\begin{proof}
  The first statement $\sigma_I[H_j] = SF_j$ is proved in \cite[Theorem 3.5]{quinn2022asymmetric}; see also \cite[Appendix B]{bal2022topological}. The second statement is proved by Lemma \ref{lem:symsigma}.
\end{proof}
For the above operator $H=D\oplus D\oplus \bar D\oplus \bar D$, then $SF(H_1)=2$ for the choice $H_1=D\oplus D$ while $SF(H_1)=0$ for the choice $H_1=D\oplus \bar D$.

\paragraph{$\Zm_2$ classification.} The classification of the operator $H$ is based on the following criterion: if it may be gapped in the interval $I$, then it is  trivial; if it cannot be gapped then it is non-trivial.  Here, we say that an operator may be gapped if it can be continuously (say in the uniform topology) deformed to an operator that admits no spectrum in any interval $[E_1,E_2]\subset (E_-,E_+)$. If the operator cannot be gapped, this implies that branches of a.c. spectrum have to cross the interval $[E_1,E_2]$ and contribute to non-trivial transport. 

This is a $\Zm_2$ invariant. We will derive the following expression for the invariant:
\begin{equation}\label{eq:Z2index}
  \ind_2\, H = (-1)^{2\pi \sigma_I[H_1]}.
\end{equation}
The index is thus given by the parity of the conductivity $2\pi\sigma_I[H_1]$ of the operator $H_1$ appearing in the decomposition \eqref{eq:spectraldecomp}. While $H_1$, and hence $2\pi\sigma_I[H_1]$, are not uniquely defined from $H$, it turns out that the parity $(-1)^{2\pi \sigma_I[H_1]}$ is indeed well defined.  The Hamiltonian is in a trivial state when $\ind_2\,H=1$ and a non-trivial state when $\ind_2\,H=-1$. 

When $H=D^{\oplus n} \oplus \bar D^{\oplus n}$, we easily verify that $\ind_2\,H = (-1)^n$ so that the index is trivial when the number of branches $n$ is even while it is non-trivial when $n$ is odd.   

The main strategy of the derivation is to write a local perturbation $\VP_{jk}$ of the spectral branches that couples two branches $j$ and $k$ so as to open a spectral gap in $I$ for the operator $H+\VP_{jk}$ restricted to these two branches. This separates operators $H$ into two categories depending on whether they have an even or an odd number of branches. We have to ensure that  $(-1)^{2\pi \sigma_I[H_1]}$ remains continuous (and hence constant) during the procedure. 


%
\paragraph{Construction of local spectral perturbations.}

For an index $j$, let $\Xi_j=(\xi_{j-},\xi_{j+})$ be the interval of support of energies $E_j(\xi)$ in $I$. 

We first consider the branches $\xi\to E_j(\xi)$ such that $E_j(\xi_{j-})=E_j(\xi_{j+})$, i.e., branches that do not cross the interval $I$ and hence $SF_j=0$ by Lemma \ref{lem:SFbranch}. Then define $\alpha_j(\xi)\in C^\infty_c(\Xi_j)$ such that $\alpha_j(\xi)+E_j(\xi)$ takes values in $I$ but does not intersect the interval $I_\delta=[E_-+\delta,E_+-\delta]$ for some $\delta>0$ small enough that $\varphi'(h)$ is supported in $I_\delta$. Such a function $\alpha_j$ may be constructed and we define 
\begin{equation}\label{eq:Vjbranches}
  \VP_j = \dint_{\Xi_j} d\xi \alpha_j(\xi) (\Pi_j(\xi) + \theta^*\Pi_j(\xi)\theta)  = \VP_{1j} + \theta^* \VP_{1j}\theta
\end{equation}
with obvious notation.  This perturbation satisfies the FTR symmetry $\theta^* \VP \theta =\VP$. The operator $H_{1\VP_j}=H_1+\VP_{1j}$ is defined and satisfies the same hypotheses [H1] as $H_1$ so that $\sigma_I[H_{1\VP_j}]$ is defined and equal to  $\sigma_I[H_{1}]$ as a direct application of Lemma \ref{lem:SFbranch}.

More generally, we observe that $\sigma_I[H_{1}+\mu \VP_j]=\sigma_I[H_{1}]$ for any $\mu\in[0,1]$ a continuous family of operators (in the uniform sense) showing that $H_1$ and $H_{1}+\VP_j$ belong to the same class as defined in \eqref{eq:Z2index}.

It thus remains to consider branches that fully cross the energy interval $I$ with $E_j(\xi_{j-})\not=E_j(\xi_{j+})$. Consider two such branches, called $E_1(\xi)$ and $E_2(\xi)$ to simplify, as well as their FTR symmetries. The branches $E_1(\xi)\Pi_1(\xi)d\xi$ and $E_1(\xi) \theta^*\Pi_1(\xi)\theta d\xi$ are both analytic in $\xi$ and with opposite spectral flows as per Lemma \ref{lem:SFbranch}.

Up to a different choice of $H_1$, we may therefore assume that $E_1$ is globally decreasing in the sense that $E_1(\xi_{1-})=E_+$ and $E_1(\xi_{1+})=E_-$ while $E_2$ is globally increasing in the same sense. We note that  $2\pi\sigma_I[H_1]$ jumps by a multiple of $2$ for such a possibly different choice of $H_1$. Thus, $(-1)^{2\pi\sigma_I[H_1]}$ is a well defined number in $\{-1,1\}$.
%

We now construct $\VP_1$ and $\VP_2$ by \eqref{eq:Vjbranches} and choose $\alpha_j\in C^\infty_c(\Xi_j)$ such that for some $\xi_j\in \Xi_j$, we have $E_j(\xi)+\alpha_j(\xi) \not\in I_\delta$ for $|\xi-\xi_j|\geq\delta$ and $E_j(\xi)+\alpha_j(\xi) = (-1)^j[(E_+-E_--2\delta)/(2\delta)]\xi+\frac12(E_++E_-)$. In other words, $E_j+\alpha_j$ does not intersect with $I_\delta$ except on an interval of size $2\delta$ where it is affine. We verify that the construction of such functions $\alpha_j$ is possible. 




Define $\tilde E_j(\xi)=E_j(\xi_j+\xi)+\alpha_j(\xi)$, $\tilde\Pi_j(\xi)=\Pi_j(\xi_j+\xi)$ as well as $\tilde\Xi_j=\Xi_j-\xi_j$, and the operators
\[
  \tilde H = H+\VP_1+\VP_2 = \dsum_{j=1}^2 \dint_{\Xi_j-\xi_j}  \tilde E_j(\xi) (\Pi_j(\xi_j+\xi) + \theta^*\Pi_j(\xi_j+\xi)\theta) = \tilde H_1 + \theta^* \tilde H_1 \theta.
\]
\begin{lemma}
  Let $H=H_1 + \theta^* H_1\theta$ and $\tilde H = \tilde H_1 + \theta^* \tilde H_1 \theta$ be the FTR-symmetric operators constructed as above. Then
  \[
    (-1)^{2\pi\sigma_I[H_1]} = (-1)^{2\pi\sigma_I[\tilde H_1]}.
  \]
\end{lemma}
\begin{proof}
  The quantity $2\pi\sigma_I[H_1]$ is well defined for $\varphi'$ supported in $(E_-+\delta,E_+-\delta)$ by assumption. Modifying the two branches $E_1(\xi)$ and $E_2(\xi)$ as described above does not modify their spectral flow so that $\sigma_I[H_1]=\sigma_I[\tilde H_1]$ and the lemma follows.
\end{proof}

We drop the  $\ \tilde{}\ $ and construct a perturbation that gaps the operator $H$ (formerly $\tilde H$). Let $\Pi_j(\xi)=\psi_j(\xi)\psi_j^*(\xi)$ since $\Pi_j(\xi)$ (formerly $\tilde \Pi_j$) is formally rank-one. Consider the operator 
\begin{align*}
  \VP_{12} &= \dint_{-\delta}^\delta \VP_{12}(\xi) d\xi\\
  \VP_{12}(\xi) &=\alpha(\xi) \big(\psi_1(\xi) \otimes \theta\psi_2(\xi)  -\psi_2(\xi) \otimes \theta\psi_1(\xi)   \big) 
    \\[1mm] & + \bar\alpha(\xi) \big( \theta\psi_2(\xi)\otimes \psi_1(\xi)  -  \theta\psi_1(\xi)\otimes \psi_2(\xi)\big).
\end{align*}
We verify that $\VP_{12}=\VP_{12}^*$ and that $\VP_{12}$ is FTR-symmetric with $\theta^* \VP_{12}\theta=\VP_{12}$. Note that $\VP_{12}(\xi)$ couples the two branches via the coupling of $\psi_j$ with the FTR symmetry of $\psi_k$ for $j\not=k$. The coupling of $\psi_j$ with $\theta\psi_j$ would generate a perturbation that is no longer FTR-symmetric. This is the origin of the symmetry-protected $\Zm_2$ invariant. It remains to choose the function $\alpha(\xi)$, which we will take real-valued.

In the basis of $\psi_1,\psi_2,\theta\psi_1,\theta\psi_2$, which is indeed a basis since these vectors live in orthogonal complements of the Hilbert space $\mH$ (more precisely, $\int \Pi_k f d\xi$ for $f\in \mH$ and $\Pi_k$ for $k=1,2,3,4$ the corresponding projectors, are orthogonal elements in $\mH$) by assumption, the operator $H+\VP_{12}$ has a density given by 
\[
  \begin{pmatrix} E_1(\xi) & 0 & 0 & \alpha(\xi) \\ 0 & E_2(\xi) & -\alpha(\xi) & 0 \\ 0 & -\bar\alpha(\xi) & E_1(\xi) & 0 \\
   \bar\alpha(\xi) & 0 & 0 & E_2(\xi)\end{pmatrix}.
\]
The eigenvalues are doubly degenerate and given by
\[
  \lambda_\pm(\xi) = \frac{E_1+E_2}2 \pm \sqrt{|\alpha|^2 + \frac{(E_1-E_2)^2}4} = \frac{E_++E_-}2 \pm \sqrt{|\alpha|^2 + \frac14 (\frac{E_+-E_--2\delta}\delta \xi)^2}.
\]
We choose $\alpha(\xi)$ such that 
\[
   \lambda_\pm(\xi) = \frac{E_++E_-}2 \pm \frac{E_+-E_--2\delta}{2} = E_\pm \mp \delta.
\]
This is achieved for
\[
   |\alpha|^2(\xi) = \Big(\frac{E_+-E_--2\delta}{2}\Big)^2 \Big( 1 - \frac{|\xi|^2}{\delta^2} \Big),
\]
so that $\alpha(\xi)$ may be chosen real-valued and non-negative, for instance. We observe that $\alpha(\xi)$ extended by $0$ outside of $[-\delta,\delta]$ is a continuous function. Standard deformations allow us to construct perturbations $\alpha(\xi)$ that are smooth and compactly supported and such that the resulting branches remain gapped in the interval $I_\delta$. 

We have therefore constructed a perturbation $\VP_1+\VP_2+\VP_{12}$ which is compact in the sense that these operators have densities that are compact support in the spectral variables $\xi$, and such that 
\[
  H_\VP= H + \VP_1+\VP_2+\VP_{12}
\]
has branches $E_1$ and $E_2$ that are now totally gapped in the interval $I_\delta$. 
\begin{lemma}
  Let $H$ and $H_\VP=\tilde H+\VP_{12}= H_{\VP1}+\theta^*H_{\VP1}\theta$ be constructed as above.  Then
  \[
    (-1)^{2\pi\sigma_I[H_1]} = (-1)^{2\pi\sigma_I[\tilde H_{\VP1}]}.
  \]
\end{lemma}
\begin{proof}
   It remains to show the result for $\tilde H$ and $\tilde H+\VP_{12}$.  In the restriction of $H$ to the two branches of interest $j=1$ and $j=2$, the operator $H_1$ is selected so that its spectral flow vanishes. This remains the case for all Hamiltonians $\tilde H+t \VP_{12}$ so that the index is indeed preserved by continuous deformations in the uniform sense (since the interval $I$ is bounded) from $H$ to $H_\VP$. 
\end{proof}

This allows us to state our main result:
\begin{proposition}\label{prop:gap} Let $H$ be an operator satisfying Hypothesis \ref{hyp:H1}.
  Let $M$ be the number of branches $1\leq j\leq M$ in the decomposition \eqref{eq:spectraldecomp}. Then there is a FTR-symmetric spectrally compact perturbation $\VP$ such that: either (i) $H_\VP=H+\VP$ is a gapped operator in the interval $I_\delta$ when $M$ is even; or (ii) $H_\VP=H+\VP$  restricted to $I_\delta$ is of the form \eqref{eq:spectraldecomp} with $M=1$. Associated to $H_\VP$ is a decomposition of the form \eqref{eq:spectraldecomp} with $H_\VP=H_{\VP1}+\theta^*H_{\VP1}\theta$. Moreover, 
\begin{equation} \label{eq:consZ2}
   \ind_2\, H : = (-1)^{2\pi \sigma_I[H_1]} =   (-1)^{2\pi \sigma_I[H_{\VP1}]} = \ind_2\, H_\VP.
\end{equation}    
\end{proposition}
\begin{proof}
  We use the above two lemmas a finite number of times to ensure that pairs of branches are gapped iteratively by continuous deformations. What remains is either a fully gapped operator or one with a unique remaining branch $E_1$. Since the transformations all preserve the $\Zm_2$ index, we obtain the result and the equality in \eqref{eq:consZ2}. More generally, we find that $H+\mu \VP$ for $\mu\in[0,1]$ generates a continuous (in the uniform topology since restricted to the bounded interval $I$) family of FTR-symmetric operators such that $(-1)^{2\pi \sigma_I[H_{1}+\mu \VP_1]}$ is independent of $\mu$.
\end{proof}

\paragraph{Construction of a Fredholm operator.}
A natural Fredholm operator \cite{PS,schulz2015z2} associated to the initial operator $H$ is given by
\begin{equation}\label{eq:FredT}
  T =  \rP U[H]  \rP  
\end{equation}
with $\rP$ a projector here (so that $\rP^2=\rP$) corresponding to multiplication by $\rP(x)\in \mathfrak{S}[0,1]$ and $U[H]=e^{2\pi i \varphi[H]}$ a unitary operator.  For operators $H$ satisfying Hypothesis \ref{hyp:H1}, $[U[H],\rP]$ is a compact operator on $\mH$ and it is then a classical result that $T$ restricted to the range of $\rP$ is a Fredholm operator; see, e.g. \cite{bal2022topological}. Equivalently, we verify that $T+\alpha(I-\rP)$ is a Fredholm operator on $\mH$ for any $\alpha\not=0$. Note that $T$ depends on $H$ only via $H_I$, its spectral restriction to the interval $I=[E_-,E_+]$ since $U(h)=1$ for $h$ outside that interval.

While the (standard) index of $T$ vanishes, we can define as in \cite{schulz2015z2} a $\Zm_2$ index, which will remain continuous (and hence constant) for the above transformations. For any bounded complex-valued function $u(h)+iv(h)$, we have
\[
  u[\theta^* H\theta] - i v [\theta^* H \theta] = \theta^* u[H] \theta - i \theta^* v[H] \theta = \theta^*(u[H]+iv[H])\theta
\]
so that $\theta^* U[H] \theta = U^*[H]$. As a consequence, we also have $\theta^* T \theta = T^*$ and the Fredholm operator is odd-symmetric. It is known \cite{schulz2015z2} that bounded odd-symmetric Fredholm operators form two distinct homotopy classes characterized by the dimension of their kernel modulo $2$.

Let $H$ and $H_\VP=H_{\VP1}+\theta^*H_{\VP1}\theta$ be as described in Proposition \ref{prop:gap}. Define 
\[
  T_\VP =  \rP U[H_\VP]  \rP \quad \mbox{ and } \quad  T_{\VP1} =  \rP U[H_{\VP1}] \rP \quad \mbox{ so that } \quad T_\VP=T_{\VP1} \oplus \theta^*T_{\VP1}\theta.
\]
After gapping an even number of branches, the resulting $H_{\VP1}$ has either one branch or none crossing $I$. It is then known \cite{bal2022topological} that the spectral flow associated to $H_{\VP1}$ is the index of the Fredholm operator $T_{\VP1}$. That index is equal to $\pm1$ when there is one branch $M=1$ in the above proposition and equal to $0$ when $M=0$. From the above decomposition, it is also clear that the dimension of the kernel of $T_\VP$ modulo $2$, called $\ind_2 T_\VP$, is equal to the index of $T_{\VP1}$ modulo 2. This computes the $\Zm_2$ index of $H$ as
\begin{equation}\label{eq:Z2indexT}
  \ind_2 H = \ind_2 H_\VP := (-1)^{\ind_2 T_\VP} = (-1)^{\ind T_{\VP1}}.
\end{equation}

We have seen in \eqref{eq:consZ2} that $(-1)^{2\pi\sigma_I[H_1]}$ defined the $\Zm_2$ invariant of $H$. For a continuous family $T(t)= \rP U[H(t)]  \rP$ of odd-symmetric Fredholm operators, the $\Zm_2$ index is independent of $t$ since  $\ind_2$ is a homotopy invariant separating odd-symmetric Fredholm operators into exactly two classes \cite{schulz2015z2}. Since $H_\VP$ is a continuous deformation of $H$ as seen in the derivation of Proposition \ref{prop:gap}, we obtain that $T_\VP$ is also a continuous deformation of $T$ so that $\ind_2\,T_\VP=\ind_2\,T$.

We thus summarize the results obtained thus far as:
\begin{theorem}\label{thm:Z2}
  Let $H$ be in the class of operators described in Proposition \ref{prop:gap}. Then the $\Zm_2$ index of $H$ defined in \eqref{eq:Z2indexT} separates that class of operators into two disjoint homotopy classes. Let $T=\rP U[H]\rP$ be a Fredholm operator (on the range of $P$) associated to $H$. The index is given explicitly by
\begin{equation}\label{eq:Z2indexsigma}
  \ind_2\, H = (-1)^{\ind_2 \,T},
\end{equation}  
where $\ind_2\, T$ is the dimension of the kernel of $T$ mod $2$.
\end{theorem}

\paragraph{Stability under perturbations.}

So far, the operators we have considered were invariant with respect to translations in $x$. Let $H$ be such an operator for which Theorem \ref{thm:Z2} applies. 

Consider a perturbation $V$ that is FTR-symmetric, i.e., $\theta^* V \theta=V$, and define $H_V=H+V$. Assume further that for $\mu\in[0,1]$, 
\begin{equation}\label{eq:Tmu}
  \mu \mapsto T_\mu = \rP e^{2\pi i \varphi(H+\mu V)} \rP
\end{equation}
is a continuous (in the uniform sense) family of Fredholm operators on the range of $\rP$. Using the stability results of the $\Zm_2$ index in \cite{schulz2015z2}, we thus deduce that 
\[
  \ind_2 T_\mu = \ind_2 T_0
\]
is independent of $\mu$. 
This provides the following definition/proposition:
\begin{proposition} \label{prop:z2stab}
  Let $H_V=H+V$ be the class of operators with $H$ as described in Theorem \ref{thm:Z2} and $V$ such that \eqref{eq:Tmu} is continuous in the uniform sense. Then we define
  \begin{equation}\label{eq:pertindex}
      \ind_2\, H_V := (-1)^{\ind_2 T_1}.
  \end{equation}
  This $\Zm_2$ invariant is given explicitly by
  \begin{equation}\label{eq:relindices}
    \ind_2\, H_V = \ind_2\, H = (-1)^{\ind_2\, T_0} = (-1)^{2\pi \sigma_I(H_1)}.
  \end{equation}
\end{proposition}
While the invariant given in \eqref{eq:pertindex} is stable, its computation based on the dimension of the kernel of $T_1$ remains a difficult task. This is the main advantage of \eqref{eq:relindices} and the decomposition involving the operator $H_1$, whose conductivity $\sigma_I[H_1]$ admits an explicit expression in a Fedosov-H\"ormander formula involving only the symbol of the operator $H_1$ \cite{bal2022topological,bal2023topological,QB22}. 

We note that while a decomposition of the form \eqref{eq:spectraldecomp} still holds with $H_{1}$ replaced by $H_{1\mu}$ (although we will not present details here), the expression of $H_{1\mu}$ and hence the computation of $\sigma_I[H_{1\mu}]$ are not explicit enough to provide a formula as compact as that for $\sigma_I[H_1]$ when $\mu=0$.

\begin{remark}\label{rem:elliptic}
    Let $H_0$ be a differential operator elliptic in the sense of \cite{bal2022topological,QB22}. Let $V(x,y)$ be a compactly supported (or more generally sufficiently rapidly decaying as $|(x,y)|\to\infty$) perturbation valued in Hermitian matrices in $\Mm(n)$ and FTR-symmetric and define the operator $V$ as pointwise multiplication by $V(x,y)$. By ellipticity assumption, $\mu\to T_\mu$ defined in \eqref{eq:Tmu} is indeed a continuous family in the uniform topology; see, e.g., \cite{bal2022topological,QB22} for a proof of such a result using a Helffer-Sj\"ostrand formulation of functional calculus. All examples shown in section \ref{sec:num} are elliptic in that sense.
 \end{remark}

\medskip

Before closing this section, we consider a few examples for which the decomposition \eqref{eq:spectraldecomp} and the assumptions [H1] hold naturally. Additional examples are presented in section \ref{sec:num}.

\paragraph{One dimensional examples.}
Recall that $D=-i\partial_x$ in one space dimension. For the operator
\[ H= D \oplus \bar D \qquad \mbox{ with } \qquad H_1 = D \begin{pmatrix}  1 & 0 \\ 0 & 0\end{pmatrix},
\]
we find that the Schwartz kernel of the projection operator density is given by
\[ \Pi_1(x,y;\xi)  = \frac{1}{2\pi} e^{i(x-y)\xi} \begin{pmatrix}  1 & 0 \\ 0 & 0\end{pmatrix},\quad \theta \Pi_1\theta^* (x,y;\xi)  = -\rJ\mK \Pi_1 \mK \rJ = \frac{1}{2\pi} e^{-i(x-y)\xi} \begin{pmatrix}  0 & 0 \\ 0 & 1\end{pmatrix}.
\] 
The two branches are associated with the (generalized) eigenvalue $E(\xi)=\xi$. We may change variables $\xi\to-\xi$ in the second branch to obtain the following spectral decomposition for the Schwartz kernel of $H$:
\[
  H(x,y) = \dint_{\Rm} \frac{1}{2\pi} e^{i(x-y)\xi} \xi  \begin{pmatrix}  1 & 0 \\ 0 & -1\end{pmatrix} d\xi.
\]
In this example, we may choose $H_1=D$ or $H_1=\bar D$. In both cases, $(-1)^{\ind_2 H}=(-1)^{2\pi\sigma_I[H_1]}=-1$ so that the material represented by such an operator is non-trivial.

This decomposition extends to the operator $H=D\oplus D \oplus \bar D \oplus \bar D \equiv  (D\oplus \bar D) \oplus  (D \oplus \bar D)$ with index  $\ind_2 \,H=(-1)^{2\pi\sigma_I[H_1]}=0$ (topologically trivial) or more generally to $H=D^{\oplus p}\oplus \bar D^{\oplus p}$ with $\Zm_2$ index given by $\ind_2\, H=(-1)^p$. This is a trivial operator for $p$ even and a non-trivial one for $p$ odd.

\paragraph{Two-dimensional example: Dirac operators.} We consider now the Dirac Hamiltonian $H=H_1\oplus \bar H_1$ with $H_1=D\cdot\sigma+y\sigma_3+V_1$ and $V_1$ a compactly supported perturbation, say. Then $H_1$ admits absolutely continuous spectrum that is independent of $V_1$ and given by an infinite number of branches $E_m(\xi)=\pm \sqrt{\xi^2+2n}$ for $m=(\pm,n)$ and $n\geq1$ as well as a branch $E_0(\xi)=-\xi$. 

Each branch $E_m(\xi)$ for $n\geq1$ appears twice, once in $H_1$ and once in $\bar H_1$. As a result these branches do not contribute to the index. Note that only finitely many such branches cross any given bounded open interval $(E_-,E_+)$; see also Fig. \ref{fig:EvsXi} below. The branch $\xi\mapsto E_0(\xi)$ provides a contribution to $H$ equivalent to $D\oplus \bar D$. As a result, we obtain that $\ind_2\,H=-1$.

A natural extension of the above model is the following direct sum of such Dirac operators
\begin{align}\label{eq:fullDiracunperturbed}
    H=h^{\oplus M}\oplus \bar{h}^{\oplus N},
\end{align}
where,
\begin{align}
    \label{eq:Dirac}
    h = D_x \sigma_1 + D_y \sigma_2 + m(y) \sigma_3  = \begin{pmatrix}
    m(y) & D_x-iD_y  \\ D_x+iD_y & -m(y)
    \end{pmatrix}.
\end{align}
 Such Hamiltonians now act on spinors of dimension $\mathcal{N}=2(M+N)$.
The above operators are FTR-symmetric when $M=N$.
 

For $V=V(x,y)$ a compactly supported FTR-symmetric perturbation (valued in $4M\times 4M$ Hermitian matrices), then $H_V=h^{\oplus M}\oplus \bar{h}^{\oplus M}+V$ satisfies the hypotheses of remark \ref{rem:elliptic}. In particular, $\ind_2\, H_V=(-1)^M$ as a direct application of Proposition \ref{prop:z2stab}.
\section{Scattering theory and {$\Zm_2$} invariant}
\label{sec:scattering}

\paragraph{Scattering theory and conductivity.}
Consider as in the preceding section, operators of the form $H=H_0+V$ with $H_0=\mF^{-1}_{\xi\to x}\hat H(\xi) \mF_{x\to\xi}$ a differential operator invariant by translation in the $x-$variable (where $\mF_{x\to\xi}$ is partial Fourier transform in the $x$ variable) and where $V$ is a multiplication operator that is slowly decaying in the $x-$variable (for concreteness assumed compactly supported in the $x-$variable). 

We wish to define a scattering theory and a scattering matrix associated to $H$  and then relate the $\Zm_2$ index to properties of the scattering matrix. In order to define the scattering matrix, we use the framework developed in \cite{chen2023scattering}.  We briefly describe what we need from that paper, to which we refer the reader for additional details. All results described below apply to the operator $H_0=h^{\oplus M}\oplus \bar{h}^{\oplus M}$  with $h$ in \eqref{eq:Dirac}.

Let $E$ be a fixed energy in the interval $[E_-,E_+]$ with $\partial_\xi E_j(\xi)\not=0$ when $E_j(\xi)=E$. Here, $E_j(\xi)$ denotes as before a branch of spectrum of $H_0$. Let $\psi$ be a generalized solution of $(H_0-E)\psi=0$, i.e., a plane wave solution existing in some $L^2_{-s}$ space  (i.e., such that $|\psi(x,y)| (1+|x|^2)^{-s}$ is square integrable, which holds for $s>\frac12$). We can then define a solution $\psi$  a generalized solution also in some  $L^2_{-s}$ space of $(H-E)\psi=0$. More specifically, 
\begin{equation}\label{eq:psiV}
  \psi^V = (I-(H-E)^{-1}_+V)\psi
\end{equation}
with $(H-E)^{-1}_+=\lim_{0<\eps\to0} (H-(E+i\eps))^{-1}$. That the latter limit exists is a consequence of a limiting absorption principle proved in \cite{chen2023scattering} for a class of operators including Dirac operators; see also \cite{ASNSP_1975_4_2_2_151_0,yamada1975eigenfunction}.

For a fixed $E$ as above, the number of linearly independent solutions $\psi$ of $(H_0-E)\psi=0$ is finite. For each such solution, we define a perturbed solution $\psi^V$ as indicated above. 

We now develop a scattering theory for the operator $H$, still following the framework in \cite{chen2023scattering}. Let $\psi_m$ and $\psi_n$ be two generalized eigenfunctions (in some $L^2_{-s}$ space) solutions of
\[
  H\psi_m = E_m \psi_m,\quad H\psi_n = E_n \psi_n
\]
with $E_n$ and $E_m$ real-valued. Define the current correlation
\[
  J_{mn}(x_0) = (\psi_n, 2\pi i[H,P(\cdot-x_0)] \psi_m ) .
\]
Such a current is defined since $\psi_{n,m}$ decay in the variable $y$ and $[H,P(x-x_0)]$ is localizing in the $x-$variable.  From \cite[Lemma 2.3]{chen2023scattering}, we deduce the current conservation:
 \begin{equation}\label{eq:currentcons}
   J'_{mn}(x_0) = 0.
 \end{equation}

We denote by $\psi_m^V$ a generalized solution of the perturbed problem $H\psi_m^V=E_m\psi_m^V$ while $\psi_m$ denotes a generalized solution of the unperturbed problem $H_0\psi_m=E_m\psi_m$.  Since $H_0$ is invariant under translations in the $x$ variable, the unperturbed solutions are of the form
\[
 \psi_m(x,y;\xi) = \frac{1}{\sqrt{2\pi}}e^{i\xi x} \phi_m(y;\xi)
\]
with $\phi_m$ solution of $\hat H(\xi)\phi_m(y;\xi) = E_m(\xi) \phi_m(y,\xi)$. We normalize $\|\phi_m(\cdot,\xi)\|=1$ and assume that $\phi_m$ decays sufficiently rapidly as $|y|\to\infty$. 

We assume that for a given $E$, there is a finite number of real-valued $\xi_m$ such that $E=E_m(\xi_m)$. We denote by $M=M(E)$ the set of corresponding indices $m$. The cardinality of $M$ is $n_++n_-$ where $n_\pm$ corresponds to the number of currents $\pm J_m>0$ associated to each unperturbed plane wave and defined as:
\begin{equation}\label{eq:currentsJn}
  J_n := \partial_\xi E_n (\xi_n)  \not=0.
\end{equation}
We thus assume that the wavenumbers $\xi_n$ are not critical points of the energy band $E_n(\xi)$. 

In the presence of a compactly supported perturbation $V$ in the $x-$variable, we have that the functions $\psi^V_m$ satisfy for $m\in M(E)$,
\begin{equation}\label{eq:approxscattering}
  \psi_m^V(x,y) \approx  \dsum_{n\in M(E)} \alpha^\pm_{mn} \psi_n(x,y) =  \dsum_{n\in M(E)} \alpha^\pm_{mn} \frac{1}{\sqrt{2\pi}} e^{i\xi_n x}  \phi_n(y) 
\end{equation}
where $a \approx b$ means that the difference $a-b$ converges to $0$ uniformly (in $(x,y)$) as $x\to\pm\infty$ and where $\alpha^\pm$ are the corresponding coefficients in these two limits. 
From this we deduce that 
\begin{equation}\label{eq:psimV}
  (\psi_m^V , 2\pi i[H,P(\cdot-x_0)] \psi_n^V ) \approx \dsum_{p,q \in M} \bar\alpha^\pm_{mp}\alpha^\pm_{nq} ( e^{i \xi_p x}\phi_p,  i[H,P(\cdot-x_0)] e^{i \xi_q x}\phi_q)
\end{equation}
where $\approx$ here is the same sense as above but now as $x_0\to\pm\infty$. All terms in \eqref{eq:psimV} are defined since $[H,P(\cdot-x_0)]$ is compactly supported in the $x-$vicinity of $x_0$. We wish to estimate the above right-hand side.
\begin{lemma}\label{lem:unperturbedcurrent}
For $P$ a switch function in $\fS(0,1)$, we have 
  \begin{equation}
    ( e^{i \xi_m x}\phi_m, i[H,P] e^{i \xi_n x}\phi_n) = \delta_{mn} \partial_\xi E_n(\xi_n) = \delta_{mn} J_n.
  \end{equation}
\end{lemma}
This is \cite[Lemma 2.4]{chen2023scattering}. We thus conclude from \eqref{eq:approxscattering} and the above lemma that in the limits $x_0\to\pm\infty$,
\begin{equation}\label{eq:decpsimn}
 (\psi_m^V,2\pi i[H,P(\cdot-x_0)]\psi_n^V) \approx \sum_p J_p \bar\alpha^\pm_{mp} \alpha^\pm_{np}.
\end{equation}

We next define the refection and transmission coefficients $R^\pm_{mn}$ and $T^\pm_{mn}$ as
\begin{align}
  \alpha^+_{mn} &= \sqrt{\frac{|J_m|}{|J_n|}}T^+_{mn} \ \ \mbox{when} \ \ J_m>0 \mbox{ and } J_n>0\\
   \alpha^-_{mn} &= \sqrt{\frac{|J_m|}{|J_n|}}T^-_{mn} \ \ \mbox{when} \ \ J_m<0 \mbox{ and } J_n<0\\
   \alpha^+_{mn} &=  \sqrt{\frac{|J_m|}{|J_n|}}R^-_{mn} \ \ \mbox{when} \ \ J_m<0 \mbox{ and } J_n>0\\
    \alpha^-_{mn} &=  \sqrt{\frac{|J_m|}{|J_n|}}R^+_{mn} \ \ \mbox{when} \ \ J_m>0 \mbox{ and } J_n<0,
\end{align}
while we also have $\alpha^-_{mm}=1$ when $J_m>0$ and $\alpha^+_{mm}=1$ when $J_m<0$. All other coefficients $\alpha^\pm_{ij}$ then vanish. We then have the following results from \cite{chen2023scattering}:
\begin{lemma}\label{lem:unitaryscattering}
 The $(n_++n_-)\times(n_++n_-)$ scattering matrix
\begin{align}\label{eq:scatteringmatrix}
      S =\begin{pmatrix} T_+ & R_- \\R_+&T_-\end{pmatrix}
\end{align}
is unitary. Here $T_+$ is the $n_+\times n_+$ matrix with coefficients $T^+_{mn}$, etc.
\end{lemma}
\begin{lemma} \label{lem:tracescattering}
   Let $S$ be the above scattering matrix. Then
\begin{equation}\label{eq:tracescattering}
  \trr\ T^*_+T_+ - \trr\ T^*_-T_- = n_+-n_-.
\end{equation}
\end{lemma}

We know from the spectral theorem that the conductivity may be directly estimated from knowledge of the generalized eigenfunctions $\psi^V_m$ as 
\[
  2\pi \sigma_I = \dsum_m   \Big| \frac{\partial\xi_m}{\partial E} \Big| (\psi_m^V , 2\pi i[H,P] \psi_m^V )
\]
and then recall the final result from \cite{chen2023scattering}:
\begin{theorem}\label{thm:sigmascattering}
 The conductivity may be recast as 
 \[
  2\pi \sigma_I =  \trr\ T^*_+T_+ - \trr\ T^*_-T_-.
 \]
\end{theorem}
\paragraph{Scattering matrix and $\Zm_2$ invariant.}

The above results hold for general operators that may not satisfy any FTR-symmetry. We now specialise to the case where $H=\theta^* H\theta$ is FTR-symmetric in which we observe then that $\sigma_I=0$ as in Lemma \ref{lem:symsigma}.

Because of the FTR symmetry, the matrices $T_\pm$ and $R_\pm$ defined in \eqref{eq:scatteringmatrix} are all $n\times n$ for some $n\geq1$ while the matrix $S$ acts on $2n-$vectors. The first $n$ components correspond to modes with positive current $(\psi,J\psi)>0$ with $J=2\pi i[H,P]$. We then observe that $(\theta\psi,J\theta\psi)<0$ so that the time-reversed modes $\theta\psi$ are the last $n$ components of such $2n-$vectors. See, e.g., \cite{fulga2011scattering} for a general classification of one-dimensional systems based on their reflection matrix.

We then have the following result:
\begin{theorem}\label{thm:scattering}
  The scattering matrix $S$ is a $2n\times 2n$ matrix and we have the following expression for the $\Zm_2$ invariant:
  \[
    \ind_2\, H = (-1)^n.
  \]
  The reflection matrices are skew-symmetric $R_\pm=-R_\pm^t$ and hence admit $0$ as an eigenvalue when $n$ is odd.
\end{theorem} 
\begin{proof}
  Let $n$ be the number of branches with positive current. We know from the preceding section that any even number of branches may be continuously deformed so that $E$ is gapped out. This shows that $\ind_2\, H = (-1)^n$ directly. 
  
  The scattering matrix is FTR-symmetric by construction of the generalized eigenvectors and scattering coefficients so that 
  \[
   \theta^* S \theta =  \theta^*  \begin{pmatrix}  T_+ & R_- \\ R_+ & T_-   \end{pmatrix} \theta =  \begin{pmatrix}  \bar T_- & - \bar R_+ \\ -\bar R_- & \bar T_+ \end{pmatrix} = S^*=\begin{pmatrix}  T_+^* &  R_+^* \\ R_-^* & T_-^* \end{pmatrix} .
  \]
  This implies that $R_\pm = -R_\pm ^t$ is skew-symmetric. When $n$, the dimension of $R_\pm$ is odd, standard results on the decomposition of skew-symmetric matrices \cite{1751-8121-41-40-405203,Y-CMS-61} show that $0$ must be an eigenvalue of $R_\pm$.
\end{proof}

For any of the operators $H_1$ such that $H=H_1+\theta^*H_1\theta$, we can construct a scattering matrix $S_1$ and scattering coefficients $T_{1\pm}$ and $R_{1\pm}$. We then have 
\[
  2\pi \sigma_I[H_1] = \Tr  (T_{1+}^*T_{1+} - T_{1-}^*T_{1-} )
\]
relating the current invariant to the scattering coefficients. We thus obtain another expression for the invariant, namely:
\begin{equation}\label{eq:ind2TR}
  \ind_2\, H = (-1)^{  \Tr  (T_{1+}^*T_{1+} - T_{1-}^*T_{1-} )}.
\end{equation}

\paragraph{Transmission and localization.} 


The invariant in the form given in Theorem \ref{thm:scattering} may be concretely computed from knowledge of the (size of the) scattering matrix $S$ at any given energy $E\in [E_-,E_+]$ that is not a critical value of the branches $E_j(\xi)$. 

A practically more useful quantity would be knowledge of $\Tr\,  T_{+}^*T_{+} = \Tr \, T_{-}^*T_{-} $ separately as these quantities provide information about transmission through the perturbation $V$.

There is no reason for such quantities to be integer-valued and stable with respect to perturbations. However, we know in several instances of topologically trivial settings that both terms converge to $0$ as the random perturbation $V$ increases. This is the regime of Anderson localization. Because of interactions with the perturbation $V$, signals are asymptotically entirely backscattered; see for instance the monograph \cite{fouque2007wave}. 

The non-trivial topology of Chern topological insulators and $\Zm_2$ topological insulators may be interpreted as an obstruction to Anderson localization \cite{PS,topological}. For FTR-symmetric Dirac operators, it is shows (for a specific model of random fluctuations) in \cite[Theorems 5.1 \& 6.3]{topological} that $\Tr\,  T_{+}^*T_{+} = \Tr \, T_{-}^*T_{-} $ is bounded below by $1$ and converges to $1$ as the thickness of the random slab $V$ increases when the $\Zm_2$ index is $-1$. In contrast, these transmission terms converge to $0$ in the case of a trivial FTR-symmetric operator. That $\Tr\,  T_{+}^*T_{+} = \Tr \, T_{-}^*T_{-} $ is bounded below by $1$ when $n$ is odd and $\ind_2=(-1)^n=-1$ is a corollary from Theorem \ref{thm:scattering}. Indeed, since $S$ is a unitary matrix, we have, among other relations,
\[
  T_{+}^*T_{+} +  R_{+}^*R_{+} = I_n.
\]
This implies that the non-negative operator $T_{+}^*T_{+}$ equals identity restricted to the kernel of $R_{+}$, which is at least one-dimensional. Thus, $\Tr\,  T_{+}^*T_{+} \geq1$. This surprising result indicates that the $\Zm_2$ invariant is indeed an obstruction to Anderson localization.

Proving that $\Tr\,  T_{+}^*T_{+}$ is asymptotically equal to $1$ in the presence of strong fluctuations is more involved and (presumably) depends on the statistics of the random perturbations; see \cite[Theorem 6.3]{topological} for such a result. The objective of the next section is to demonstrate such results numerically for more complex models than Dirac operators and for a large class of perturbations $V$.

\section{Integral formulation and computation of scattering matrices}
\label{sec:num}
We consider differential models where the theory presented in the preceding sections applies and construct an integral formulation for the numerical simulation of perturbed systems that allows us to compute scattering matrices robustly and accurately.

Our first class of models is based on systems of Dirac equations and on the defining building block $h$, a $2\times2$ system defined in \eqref{eq:Dirac}.
The unperturbed generalized Dirac model is then described by an $\mathcal{N}\times\mathcal{N}$ Hermitian matrix-valued operator of the form 
\[H=h^{\oplus M}\oplus \bar{h}^{\oplus N}, \qquad \mathcal{N}=2(M+N).\] 
We use $H_V=H+V$ to denote the perturbed system where $V$ models some Hermitian perturbation.
To simplify calculations, we perform a unitary transformation replacing building block $h$ in \eqref{eq:Dirac} and \eqref{eq:fullDiracunperturbed} by $qhq^*$, which we will still call $h$, for $q=\frac{1}{\sqrt2} (\sigma_1+\sigma_3)$. We observe that $q^*=q$ and $q(\sigma_1,\sigma_2,\sigma_3)q=(\sigma_3,-\sigma_2,\sigma_1).$  In this new basis, 
\begin{align} \label{eq:DiracNew}
     h = D_x \sigma_3 - D_y\sigma_2 + y \sigma_1. 
\end{align}
We also still denote by $V$ the transformed perturbation $QVQ^*$, where $Q=q^{\oplus(m+n)}$ so that $H_V$ is now expressed in the new basis as well.

In order to obtain explicit expressions for Green's functions (Schwartz kernels of resolvent operators), we assume $m(y)=y$ a linear domain wall as in \cite{bal2023asymmetric}. 

\subsection{Generalized Dirac operators and spectral analysis}

The unperturbed operator $h$ in \eqref{eq:DiracNew} is invariant under translations in the $x-$variable. Denoting by $\xi$ the dual Fourier variable to $x$ and $\mF$ the corresponding transform, we observe that $h=\mF^{-1}\hat h(\xi)\mF$ satisfies
\begin{align} \label{eq:hfourier}
   \hat h(\xi)-E = \xi\sigma_3 - D_y\sigma_2 + y \sigma_1 -E=\begin{pmatrix}  \xi_m-E & \fa \\  \fa^* & -(\xi_m+E) \end{pmatrix},
\end{align}
where $\fa=\partial_y+y$ may be interpreted as an annihilation operator. Indeed, since $\fa e^{-\frac12 y^2}=0$, the $L^2(\Rm)-$kernel of $\fa$ is one-dimensional while that of $\fa^*=-\partial_y+y$ is trivial.

We also consider generalizations of the above operators with a building block $\hat h_p(\xi)$ replacing $\hat h(\xi)\equiv \hat h_1(\xi)$ and defined by
\begin{align} \label{eq:hpfourier}
   \hat h_p(\xi)-E =\begin{pmatrix}  \xi-E & \fa^p \\  (\fa^*)^p & -(\xi+E) \end{pmatrix}.
\end{align} 
To simplify notation, we often omit the index $p$ in what follows.

Following a similar procedure as the one proposed in \cite{topological,bal2023asymmetric}, we now construct the (never trivial) kernel of $h-E$ for $E\in R$. To this end, we first define the countable set $M_h$ as the union of the following (pairs of) indices $m$. For $\Nm\ni n\geq p$, we define $m=(n,\epm)$ with $\epm=\pm1$, while for $0\leq n<p$, we define $m=(n,-1)$. In particular, $(n,1)\not\in M_h$ when $n<p$. This asymmetry reflects the fact that the $L^2-$kernel of $\fa^p$ has dimension $p$ while that of $\fa^*$ is trivial.

The kernel of $h-E$ is then found to be 
\begin{align}
    \psi_{h,m}(x,y;E)=e^{i\xi_m x}\phi_{h,m}(y;E),
\end{align}
where 
\begin{equation}\label{eq:xim}
    \xi_m=\left\{\begin{aligned}
        &\epm(E^2-2^p\frac{n!}{(n-p)!})^{\frac12}\quad &n\geq p\\
        &\epm E \quad &n< p
    \end{aligned}\right.
\end{equation}
and
\begin{align}\label{def:phi}
    \phi_{h,m} = c_m \begin{pmatrix} \fa^p \varphi_n \\ (E-\xi_m) \varphi_n \end{pmatrix} =  c_m \begin{pmatrix} \beta_n \varphi_{n-p} \\ (E-\xi_m) \varphi_n \end{pmatrix} ,\qquad c_m^{-2} = \beta_n^2 + |E-\xi_m|^2.
\end{align}
Here, $\varphi_n$ denote the $n$-th Hermite polynomials, and  $\beta_n^2=2^p\frac{n!}{(n-p)!}$. In particular, $\beta_n=0$ when $n<p$. The spectral decomposition of $\bar{h}$ is similar. We find:
\begin{align}
\psi_{\bar h,m}=e^{i\xi_m}\phi_{\bar h,m}(y;E),
\end{align}
where $\phi_{\bar h,m}(y;E)= c_{\bar h,m}\begin{pmatrix} \fa^p \varphi_n \\ (E+\xi_m) \varphi_n \end{pmatrix} =  c_{\bar h,m} \begin{pmatrix}  \beta_n \varphi_{n-1} \\ (E+\xi_m) \varphi_n \end{pmatrix}$ with $ c_{\bar h,m}^{-2} = \beta_n^2 + |E+\xi_m|^2$, is the solution of,
\begin{align} \nonumber
  (\hat{\bar{h}}(\xi_m)-E) \phi_{\bar h,m} =  \begin{pmatrix}  -\xi_m-E & \fa^p \\  (\fa^*)^p & \xi_m-E \end{pmatrix} \phi_{\bar h,m}=0.
\end{align}
Here, due to the time-reversal symmetry, we assume $m=(n,\epm)$ with $\epm=\pm1$ for $n\geq p$, while for $n<p$, we define $m=(n,+1)$. We denote such index set of $m$ for $\bar h$ as $M_{\bar h}$

We observe that the functions $\{\phi_{h,m}(y;E)|m\in M_h\}$ (or $\{\phi_{\bar h,m}(y;E)|m\in M_{\bar h}\}$) form a complete family and hence a basis of $L^2(\Rm_y;\Cm^{2})$.
The eigenspaces of the unperturbed operator $H=h^{\oplus M}\oplus \bar{h}^{\oplus N}$ may be decomposed as a direct sum of functions $\{\phi_{h,m}(y;E)|m\in M_h\}$ and $\{\phi_{\bar h,m}(y;E)|m\in M_{\bar h}\}$. By a slight abuse of the notation for index $m$, we denote $\psi_m(y;E)$ as the $m$-branch solution of $(H-E)\psi=0$ with   
\begin{align}\label{def:psimfull}
    \psi_{m}(x,y;E)=e^{i\xi_m x}\phi_{m}(y;E)
\end{align}
and $m=(\sm,n,\epm)\in \mathcal{M}=(M_h)^{\oplus M}\oplus(M_{\bar h})^{\oplus N}$. More precisely, for $s_m\leq M$, $m$ denotes a branch corresponding to the $(n,\epm)$ branch of the $p$-th operator $h$ in the full operator $h^{\oplus M}$. Then $\phi_{m}(y;E)=e_\sm\otimes \phi_{(n,\epm)}$ where $e_{\sm}$ is the $\sm$ coordinate vector of $\Rm^{M+N}$. Note that when $n=0$, then $\epm<0$. For $\sm> M$, $m$ denotes a branch corresponding to the $(n,\epm)$ branch of the $(\sm-M)$-th operator $\bar h$ in the full operator ${\bar h}^{\oplus M}$. Then $\phi_{m}(y;E)=e_\sm\otimes \phi_{(n,\epm)}$ where $e_\sm$ is the $\sm$ coordinate vector of $\Rm^{M+N}$. Recall that $\epm>0$ when $n< p$.  

\subsection{Green's Function of $H-E$ and integral formulation}

 The unperturbed {\em outgoing} Green's kernel for $h-E$ is constructed explicit in \cite{bal2023asymmetric}. Since $h$ admits the whole real line as (absolutely) continuous spectrum, the operator $h-E$ is not uniquely boundedly invertible. The outgoing Green's function for $E\in\Rm$ is the kernel of the well-defined operator $\lim_{\omega\to0}(h-(E+i\omega))^{-1}$; see \cite{bal2023asymmetric}. For generalized Dirac operators $H-E$, we first list the Green's kernels $G_h$ and $G_{\bar h}$ of $h-E$ and $\bar h -E$, respectively, as:
 \begin{align}\label{green_function_matrix}
   G_h =  \begin{pmatrix} (D_x+E) G_+ & \fa^p G_- \\ (\fa^*)^p G_+ & (-D_x+E) G_- \end{pmatrix}, \\
   G_{\bar h} =  \begin{pmatrix} (E-D_x) G_+ & \fa^p G_- \\ (\fa^*)^p G_+ & (D_x+E) G_- \end{pmatrix}\label{green_function_matrix2},
\end{align}
where 
\begin{align}
\label{G_plus}
   G_+(x,y;y_0) = \dsum_{n\geq0} \frac{-1}{2\theta_{n+p}} e^{\theta_{n+p}|x|} \varphi_n(y) \varphi_n(y_0),\\
\label{G_minus}
    G_-(x,y;y_0) = \dsum_{n\geq0} \frac{-1}{2\theta_n} e^{\theta_n|x|} \varphi_n(y) \varphi_n(y_0),
\end{align}
and $ \theta_n=  i\sqrt{E^2-\beta_n^2} =i\xi_{(n,1)}$.

We then construct the Green's function for $H-E$ using the building blocks in \eqref{green_function_matrix} and \eqref{green_function_matrix2} as follows:
\begin{align}\label{eq:FullGreensKernel}
    (H-E)^{-1}=(h^{\oplus M}-E)^{-1}\oplus (\bar{h}^{\oplus N}-E)^{-1}=(G_h)^{\oplus M}\oplus (G_{\bar h})^{\oplus N}.
\end{align}

Following the method developed in \cite{bal2023asymmetric}, we use the above unperturbed Green's function to solve for generalized eigenfunctions of the perturbed operator $H_V=H+V$. We highlight the differences and refer the interested readers to \cite{bal2023asymmetric} for a more detail discussion.

The solution of perturbed eigenproblem $H_V$ is represented by 
\begin{align} \nonumber
    \psi=\psin+\psiout,
\end{align} 
where $\psin$ is one of the possible incoming plane waves, e.g.,  $\psi_m$ described in \eqref{def:psimfull}, and $\psiout$ is the corresponding outgoing solution. We denoted $\psin+\psiout$ as $\psi^V$ in \eqref{eq:psiV}.  As in \cite{bal2023asymmetric}, $\psiout$ admits a convenient integral representation via the Green's kernel of the unperturbed problem
\begin{align}\label{representation_psiout}
   \psiout(x,y)= \int_{\Rm^2} G(x,y;x_0,y_0)\rho(x_0,y_0)dx_0dy_0.
\end{align}
The density $\rho$ we wish to compute then satisfies the following integral equation
\begin{align}\label{IntegralFormulation}
\rho(x,y)+V(x,y)\int_{\Rm^2} G(x,y;x_0,y_0)\rho(x_0,y_0)dx_0dy_0=-V(x,y)\psin(x,y).
\end{align}

In what follows, we assume that $V$ is a smooth Hermitian matrix-valued compactly supported function. As is apparent from the above formulation, the density $\rho$ vanishes outside of the support of $V$. Therefore, only the support of $V$ needs to be discretized in the numerical simulations.  
The above formulation for $\rho$ is therefore much less expensive both memory-wise and in computing time than a method computing $\psiout$ directly.

\subsection{Computation of scattering matrices}\label{sec:algo}
For $V$ compactly supported, we decompose the unknown density $\rho$ as: 
\begin{align}\label{rho_coeff}
    \rho(x,y)=\sum_{i,n} \rho_{i,n}
P_i(x)\varphi_n(y),
\end{align}
where $P_i$ denotes the i-th Legendre polynomial on the support of $V$,  $\varphi_n$ is $n$-th order Hermite functions, and $\rho_{i,n}\in \Cm^{\mathcal{N}}$ are the corresponding generalized Fourier coefficients. 
The above sum is truncated as 
\begin{align}\label{rho_truncated}
    \tilde{\rho}(x,y)=\sum_{i\le n_x,n\le n_y} \rho_{i,n}
P_i(x)\varphi_n(y).
\end{align}
We refer the reader to \cite{bal2023asymmetric} for details of the implementation and computation of \eqref{rho_truncated}.

Once the truncated density $\tilde{\rho}$ is constructed on $I\times \Rm$ for some interval $I=[x_L,x_R]$, we can recover the solution of $H_V-E$ on the whole physical space by
\begin{align}
    \psi(x,y) = \psin(x,y) + \int_{\Rm^2} G(x,y,x_0,y_0) \tilde{\rho}(x_0,y_0)dx_0dy_0.
\end{align}

Since $\{\phi_{h,m}(y;E)|m\in M_h\}$ (or $\{\phi_{\bar h,m}(y;E)|m\in M_{\bar h}\}$) form a basis of $L^2(\Rm_y;\Cm^{2})$, we can recast $\psi$ as a combination of branches of kernels of the unperturbed operator $H-E$ as,
\begin{align}\label{psi_decomposition}
	\psi(x,y) = \dsum_{m\in \mathcal{M}}  \alpha_m(x)  \psi_m(x,y),
\end{align}
with $\alpha_m(x)=(\psi_m(y),\psi(x,y))_y$  constant outside of $I$.
 
Instead of mapping incoming fields to outgoing fields (in the far field), we consider the corresponding linear map on $[x_L, x_R]$.  The incoming condition is then the coefficients of the right traveling modes $\phi_m$ ($\epm>0$) at the left boundary and the left traveling modes $\phi_m$ ($\epm<0$) at the right boundary, namely  $\alpha_{+}(x_L)\in \Cm^{M(n_y-p)+Nn_y}$ and $\alpha_{-}(x_R)\in\Cm^{Mn_y+N(n_y-p)}$. The outgoing solution consists of the coefficients of the left traveling modes $\phi_m$,($\epm<0$), at the left boundary and the right traveling modes $\phi_m$ ($\epm>0$) at the right boundary, namely  $\alpha_{-}(x_L)$ and $\alpha_{+}(x_R)$. 

The transmission (TR) matrix $\mathcal{S}$ corresponding to the perturbation $V$ and computed on the interval $[x_L,x_R]$, is then defined as,
\begin{align}\label{def:TRMatrix}
    \begin{pmatrix}
     \alpha_{+}(x_R)\\
    \alpha_{-}(x_L)\\
    \end{pmatrix}=
    \mathcal{S}
    \begin{pmatrix}
        \alpha_{+}(x_L)\\
    \alpha_{-}(x_R)
    \end{pmatrix}, \qquad \qquad \mathcal{S} = \begin{pmatrix}
T_+ & R_-\\ R_+ & T_-
    \end{pmatrix}.
\end{align}
We point out that the transmission reflection matrix $\mathcal{S}$ in \eqref{def:TRMatrix} is a generalization of the scattering matrix $S$ in \eqref{eq:scatteringmatrix}. Indeed, in the numerical computations, we need to include evanescent modes ($m\in M_h$ or $M_{\bar h}$ for which $\xi_m$ is purely imaginary). As discussed in Section \ref{sec:scattering}, the evanescent modes do not contribute to currents and the index.  However, they are essential in constructing a merge-and-split algorithm to compute the scattering matrix for perturbations $V$ supported on a larger domain.
For instance, let $  \mathcal{S}^1=  \begin{pmatrix}
T_+^1 & R_-^1\\ R_+^1 & T_-^1
    \end{pmatrix}$ and $  \mathcal{S}^2=  \begin{pmatrix}
T_+^2 & R_-^2\\ R_+^2 & T_-^2
    \end{pmatrix}$ denote the TR matrices of two adjacent intervals in the $x$ direction, $I_1$ and $I_2$, respectively. Then the TR matrix for $I_1\cup I_2$ is 
\begin{align}
    \begin{pmatrix}
   T_+^1(I-R_-^1R_+^1)^{-1}T_+^1 &T_+^1(I-R_-^1R_+^1)^{-1}R_-^1T_-^1 +R_-^1\\
  T_-^1(I-R_+^1R_-^1)^{-1}R_+^1T_+^1+R_+^1 & T_-^2(I-R_+^1R_-^1)^{-1}T_-^1
    \end{pmatrix}\label{scattermatrix_formula}.
\end{align}

We apply a {\em binary merging} strategy to compute the transmission reflection matrix over long intervals. For an algorithm with $L$ levels of merging  over interval $[0,l]$, we divide $[0,l]$ into $2^L$ disjoint intervals of length $\frac{l}{2^L}$. On each interval (leaf), we compute the transmission reflection matrices in \eqref{def:TRMatrix}. We next apply the merging formula \eqref{scattermatrix_formula} to adjacent intervals $[\frac{2kl}{2^L},\frac{(2k+1)l}{2^L}]$ and $[\frac{(2k+1)l}{2^L},\frac{(2k+2)l}{2^L}]$ for $k=0,\ldots, 2^{L-1}-1$. We then apply the merging formula \eqref{scattermatrix_formula} to compute the TR matrix on $2^{L-2}$ intervals from those computed on $2^{L-1}$. We repeat this procedure iteratively until we reach the TR matrix for the full interval $[0,l]$.  Details of the algorithm may be found in \cite{bal2023asymmetric}. Once the TR matrix $\mathcal{S}$ is constructed on a final interval, the scattering matrix $S$ on that interval may be evaluated by limiting $\mathcal{S}$ to entries that only relate to the propagating modes. This procedure was used in \cite{bal2023asymmetric} to compute scattering matrices essentially to arbitrary precision (typically with at least ten digits of accuracy).

\section{Numerical Examples}
\label{sec:examples}
In this section, we present numerical simulations of the generalized Dirac models introduced in section \ref{sec:num} using the integral formulation of section \ref{sec:algo}. After validating the accuracy and convergence properties of the algorithm, our main objective is to compute transmission coefficients $\Tr(T^*_+T_+)$ and $\Tr(T^*_-T_-)$ as presented in section \ref{sec:scattering} and estimate how they depend on the perturbation $V$ and in particular on the size of the support of $V$. 
\subsection{Verification of Convergence}\label{sec:convergence}
We first numerically validate the convergence of the algorithm of section \ref{sec:algo} by computing the scattering matrix $S$ for increasing values of the discretization parameters $(n_x,n_y)$ in \eqref{rho_truncated}. 

Consider the unperturbed system $H=(h\oplus \bar h)^{\oplus 2}$ with thus $M=N=2$ and $p=1$ and acting on an $8$-vector. In Fig. \ref{fig:EvsXi}, we display  the first few branches of continuous spectrum of $h$ (left) and $\bar h$ (right). The two spectra are symmetrical (in $\xi\to-\xi$) except for the non-dispersive (linear) branch characterized by opposite group velocities $-1$ and $+1$ for $h$ and $\bar h$, respectively.

For a choice of energy $E=1.8$, we observe from Fig. \ref{fig:EvsXi} that each building block $h$ in $H$ has exactly two left-propagating modes and one right-propagating mode.
The time-reversed counterpart $\bar h$ thus has two right-propagating modes and one left-propagating mode.  
\begin{figure}[ht!]
    \centering
      \begin{subfigure}{0.45\textwidth}
		\includegraphics[width = \linewidth]{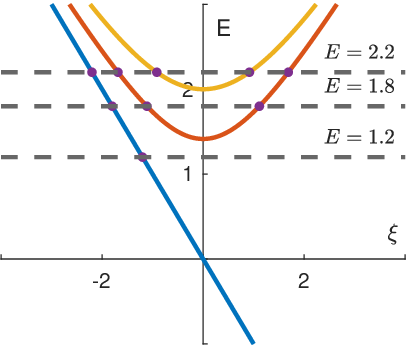}
		\caption{Branches of $h$}
  \end{subfigure}
\begin{subfigure}{0.45\textwidth}
		\includegraphics[width = \linewidth]{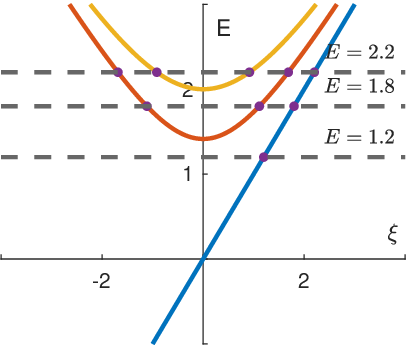}
		\caption{Branches of $\bar h$}
  \end{subfigure}
\caption{Branches of continuous spectrum of $h$ and $\bar h$.}
    \label{fig:EvsXi}
\end{figure}
The resulting scattering matrix of $H$ at $E=1.8$ is therefore a $12\times 12$ matrix with each component $T_\pm$ and $R_\pm$ being $6\times6$ matrices.

Introduce now a perturbation $V_1$ defined as 
\begin{align}
    V_1(x,y;l)=v_1(x,y)\Big(I_8+\begin{bmatrix}
        0 &1\\ 1&0
    \end{bmatrix}\otimes\begin{bmatrix}
        1 &1\\ 1&1
    \end{bmatrix}\otimes \sigma_2\Big)\chi_{[0,l]}
\end{align} where
\begin{align}\label{v_p1}
    v_1(x,y)=\exp(-y^2)\big(&y\cos((-E-\sqrt{E^2-2})x)+y\cos((-E+\sqrt{E^2-2})x)\\ \nonumber
    &+\cos(2\sqrt{E^2-2}x)+\cos(2Ex)\big).
\end{align}
The oscillatory structure of $V_1$ is introduced to maximize interactions/coupling between the different propagating modes.  We take the length of the perturbation as $l=1$. To compute the  reference scattering matrix, we take $(n_{x,ref},n_{y,ref})=(24,150)$. In  Fig.\ref{fig:convergence_leaf}, we compare the relative error of scattering matrix in Frobenius norm for increasing choices of $(n_x,x_y)$.
\begin{figure}[ht!]
    \centering
      \begin{subfigure}{0.32\textwidth}
		\includegraphics[width = \linewidth]{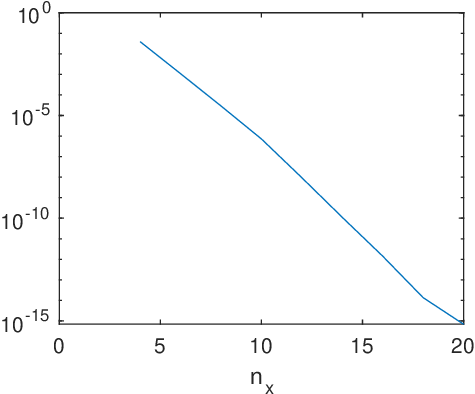}
		\caption{varying $n_x$, $n_y=n_{y,ref}=150$}
  \end{subfigure}
\begin{subfigure}{0.32\textwidth}
		\includegraphics[width = \linewidth]{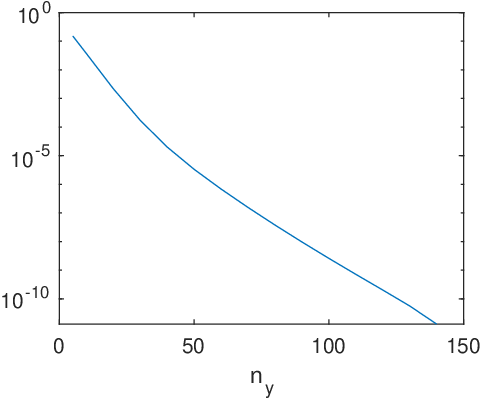}
		\caption{varying $n_y$, $n_x=n_{x,ref}=24$}
  \end{subfigure}
\caption{Convergence of scattering matrix in Frobenius norm without merging under various discretization}
    \label{fig:convergence_leaf}
\end{figure}
We observe an exponential convergence to the solution as $(n_x,n_y)$ increase as expected in a situation with smooth coefficients $V$ and consistent with the simulations in \cite{bal2023asymmetric}.

We next turn to the validation of the convergence of the algorithm with binary merging strategy. We consider the same configuration as for Fig.\ref{fig:convergence_leaf}.  We split the interval $[0,l]$ into $2^L$ equal-length sub-intervals. The scattering matrix computed by  $(n_{x,ref},n_{y,ref})=(24,150)$ without merging as in the previous example is taken as the ground truth reference. 

\begin{figure}[ht!]
    \centering
      \begin{subfigure}{0.32\textwidth}
		\includegraphics[width = \linewidth]{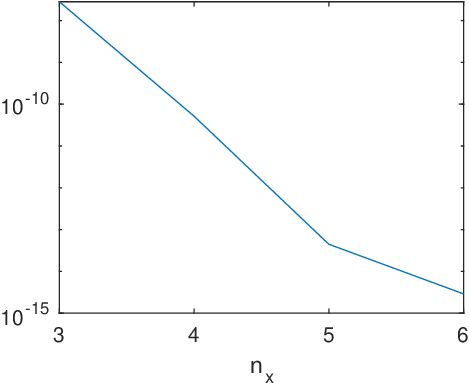}
		\caption{varying $n_x$, $L=4$, $n_y=150$}
  \end{subfigure}
\begin{subfigure}{0.32\textwidth}
		\includegraphics[width = \linewidth]{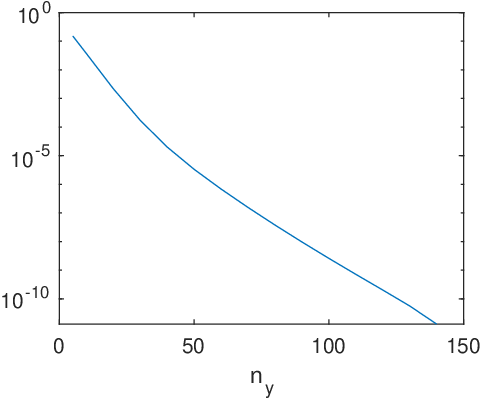}
		\caption{varying $n_y$, $L=4$,  $n_x=8$}
  \end{subfigure}
  \begin{subfigure}{0.32\textwidth}
		\includegraphics[width = \linewidth]{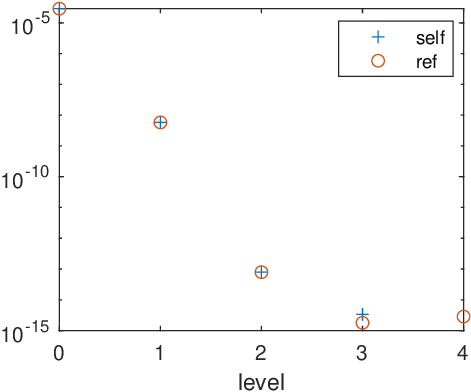}
		\caption{varying $L$, $n_x=8$, $n_y=150$}
  \end{subfigure}

\caption{Convergence of scattering matrix in Frobenius norm with the merge-and-split algorithm}
    \label{fig:convergencefull}
\end{figure}
Fig.\ref{fig:convergencefull} displays errors in the scattering matrix in Frobenius norm for increasing values of $(n_x,n_y,L)$. From Fig.\ref{fig:convergencefull}, we observe that with $L=4$ levels of merging, we need only to select $n_x=6$ to achieve a similar high accuracy (near machine-precision accuracy $10^{-15})$ as one with $n_x=20$ in Fig.\ref{fig:convergence_leaf}(a). In the numerical simulations below, we always select $L$ so that the length of each leaf interval is no greater than  $\frac{1}{2^4}$. Fig.\ref{fig:convergencefull}(b) displays the dependence on $n_y$, which is similar to that observed in Fig.\ref{fig:convergence_leaf}(b) since merging only occurs in the direction $x$. Fig.\ref{fig:convergencefull}(c) displays computational errors for levels of merging between one and four when compared to either the ground truth reference without merging (denoted as \emph{ref}) or the reference obtained with four levels of merging (denoted as \emph{self}). This shows that with a high level of merging, the algorithm is accurate even for a small value of $n_x$.

\subsection{FTR symmetry and Anderson Localization for Dirac operators}\label{sec:4x4TRS}
We now illustrate with numerical simulations one of the main features of FTR-symmetric Hamiltonians, namely a robust transmission in both directions even in the presence of strong FTR-symmetric fluctuations in the system. This manifestation may be seen as a topologically protected obstruction to Anderson localization. The latter, which requires transmission to decay as the length of a slab of perturbations increases, is confirmed numerically for topologically trivial FTR-symmetric Hamiltonians as well as for nontrivial FTR-symmetric Hamiltonians in the presence of non-FTR-symmetric perturbations.


We start with an unperturbed operator given by $H=(h\oplus\bar h)$ acting on a four-dimensional spinor with thus $p=1$ as well as $M=N=1$. As shown in Lemma \ref{lem:symsigma}, the conductivity $2\pi\sigma_I[H+V]$ of the system is $0$ regardless of the (say, spatially compactly supported) perturbation $V$. 

However, since $2\pi\sigma_I=\Tr (T^*_+T_+)-\Tr (T^*_-T_-)=0$, each individual term $\Tr (T^*_+T_+)=\Tr (T^*_-T_-)$ referred to as one-sided transmissions does not necessarily vanish. In the absence of perturbation $V$, these terms simply count the number of propagating modes in the system. This number is energy-dependent and for instance equal to $3$ for $E=1.8$ as already observed in Fig. \ref{fig:EvsXi}. 

The theory of section \ref{sec:scattering} shows that $\ind_2\, H=\ind_2\,H_V=-1$, i.e., $H$ models a topologically non-trivial FTR-symmetric phase. As a consequence, $H$ cannot be gapped and we thus expect robust transport along the $x-$axis even in the presence of FTR-symmetric perturbations. We have shown that $\Tr (T^*_+T_+)\geq1$ in such a setting and the results obtained in \cite{topological} show that $\Tr (T^*_+T_+)$ is close to $1$ in the presence of strong randomness. At the same time, the fact that $2\pi\sigma_I[H+V]=0$ indicates that Anderson localization (characterized by asymptotically small transmission across a large slab of randomness) should prevail for non-FRT-symmetric perturbations $V$. This is what we aim to demonstrate numerically. 

We consider two types of perturbation $V$ in \eqref{v_p1}. The first one is FRT-symmetric:
\begin{align}\label{eq:V2TRS}
    V_{TR}(x,y;l)=v_1(x,y)\begin{bmatrix}
    I_2 & (1-i)\sigma_2 \\ (1+i)\sigma_2 & I_2
\end{bmatrix}\chi_{[0,l]}.
\end{align}
The second perturbation is a slight modification of \eqref{eq:V2TRS} that breaks the FTR symmetry: 
\begin{align}\label{eq:V2NTR}
    V_{NTR}(x,y;l)=v_1(x,y)\begin{bmatrix}
    I_2 & (1-i)\sigma_2 \\ (1+i)\sigma_2 & -I_2 
\end{bmatrix}\chi_{[0,l]}.
\end{align}
Consider an energy level $E=1.8$ and $(n_x,n_y)=(6,40)$ in \eqref{rho_truncated} to achieve high accuracy in the simulations. We apply $L=9$ levels of binary merging on the interval $[0,32]$ to compute the scattering matrices and one-sided transmissions. The levels of merging $L$ is selected to achieve machine-precision accuracy in $x$ direction, see Figure.\ref{fig:convergencefull}(c) and related discussion.

Fig.\ref{fig:TR_symmetry}(a)-(b) presents the computed one-sided transmissions $\Tr  (T_+^*T_+)$ and $-\Tr( T_-^*T_-)$ against the length of the perturbation $l$ for the two perturbations defined above.

\begin{figure}[ht!]
    \centering
      \begin{subfigure}{0.32\textwidth}
		\includegraphics[width = \linewidth]{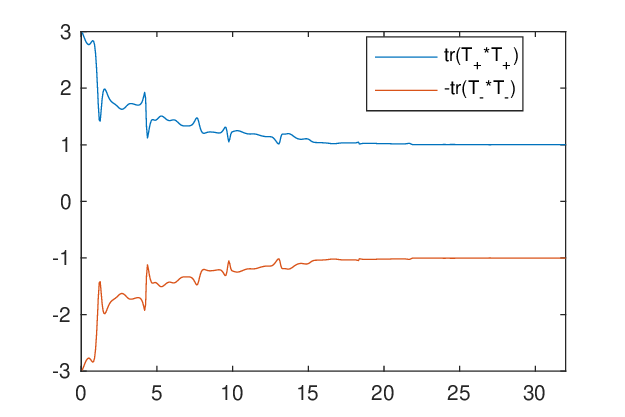}
		\caption{FTR-Symmetry preserving $V_{TRS}$ in  \eqref{eq:V2TRS}}
  \end{subfigure}
\begin{subfigure}{0.32\textwidth}
		\includegraphics[width = \linewidth]{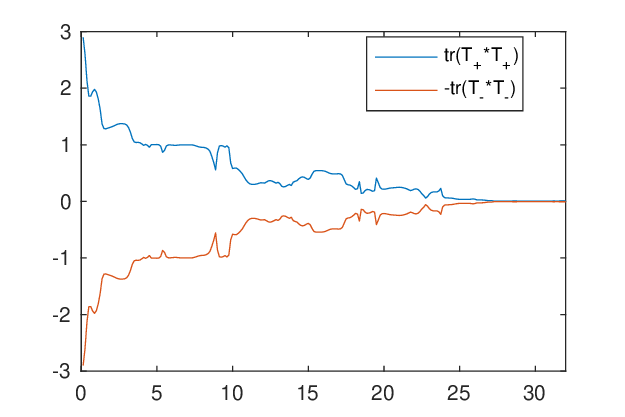}
		\caption{FTR-Symmetry breaking  $V_{NTR}$ in  \eqref{eq:V2NTR}}
  \end{subfigure}
  \begin{subfigure}{0.32\textwidth}
		\includegraphics[width = \linewidth]{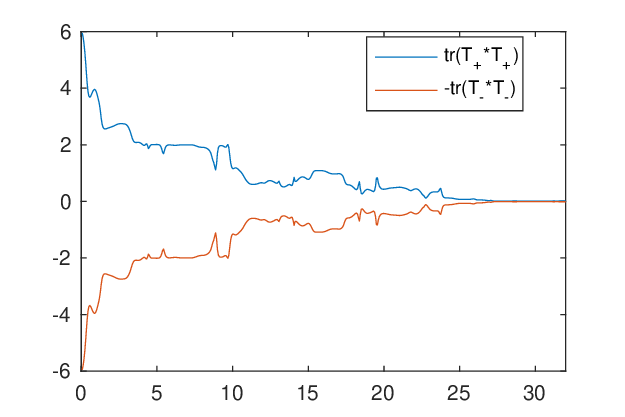}
		\caption{FTR-Symmetry preserving  $V_{TRS,M=2}$ in  \eqref{eq:V28x8} with $M=2$ }
  \end{subfigure}
\caption{One-sided transmissions against the length of the perturbation support. Blue: $\Tr  (T_+^*T_+)$, Red: $-\Tr( T_-^*T_-)$. }
    \label{fig:TR_symmetry}
\end{figure}

In both cases, the interface conductivity $2\pi\sigma_I[H+V]=\Tr (T^*_+T_+)-\Tr( T^*_-T_-)$ clearly vanishes (with machine-precision accuracy) regardless of the perturbation, which is consistent with Lemma \ref{lem:symsigma}. More interestingly, $\Tr( T^*_+T_+)$, which equals $3$ when $V=0$, is no longer constant and depends on the perturbation $V$. We note that $\Tr( T^*_+T_+)\geq1$ for $V$ satisfying a FRT-symmetry. Moreover, as expected from the results in \cite{topological}, we observe that $\Tr( T^*_+T_+)$ converges to $1$ as $l$ increases for the FRT-symmetric $V$ while it converges to $0$ for the more general non-FRT-symmetric $V$. This confirms the $\Zm_2$ index as a (partial, quantized) obstruction to Anderson localization (all modes  but one localize in the non-trivial $\Zm_2$ case). In the non-FRT-symmetric $V$, full Anderson localization is numerically observed as one-sided transmissions converge to $0$ in the presence of strong randomness. 

\medskip

Let us now consider the operator $H_V=(h+\bar h)^{\oplus2}+V_{TRS,2}$ corresponding to $M=N=2$ while still $p=1$. The FTR-symmetric perturbation is given by:
\begin{align}\label{eq:V28x8}
     V_{TRS,M=2}(x,y;l)=v_1(x,y)\begin{bmatrix}
    I_2 & & & (1-i)\sigma_2 \\ 
    & -I_2&  (1-i)\sigma_2 &   \\
    & (1+i)\sigma_2& I_2 &\\
   (1+i)\sigma_2 & & & -I_2
\end{bmatrix}\chi_{[0,l]}.
\end{align}
Since $M$ is even, we obtain from earlier sections that $\ind_2\, H=\ind_2\, H_V=1$ and that the Hamiltonian is thus topologically trivial.
We present the one-sided transmissions in Fig.\ref{fig:TR_symmetry}(c). We observe that these transmissions decay to $0$ as randomness increases, which is consistent with (full) Anderson localization.

We also note that the transmissions have the same profile in Fig.\ref{fig:TR_symmetry}(b) and (c), with the latter being equal to twice the former. This is not a coincidence. The operator $H_V=(h\oplus\bar h)^{\oplus2}+V_{TRS,2}$ is equivalent to $((h\oplus\bar h)+V_{NTR})\oplus ((h\oplus\bar h)-\overline{ V_{NTR}})$, by implementing the unitary transformation 
\begin{align}
    h_1\oplus \bar h_1 \oplus h_2 \oplus \bar h_2 \to h_1\oplus \bar h_2 \oplus h_2 \oplus \bar h_1.
\end{align}
It remains to verify, which we did numerically, that $(h\oplus\bar h)-\overline{ V_{NTR}}$ and $(h\oplus\bar h)+V_{NTR}$ have the same one-sided transmission $\Tr (T^*_+T_+)$.

The comparison between \eqref{eq:V2TRS}, \eqref{eq:V2NTR} and \eqref{eq:V28x8} is enlightening. While $h_1$ and $\bar h_1$ can only be coupled by a FTR-non-symmetric perturbation, a FTR-symmetric extension of that perturbation may be used to couple $h_1$ and $\bar h_2$ as well as $h_2$ and $\bar h_1$. This coupling is at the core of the classification presented in earlier sections.

\subsection{More general cases of FTR-symmetric operators}
We now simulate scattering matrices of more general Dirac models and investigate the asymptotic behavior of the one-sided transmissions.
For a generalized Dirac system $H_V=h^{\oplus M}\oplus\bar{h}^{\oplus M}+V$, then $H_V$ is FTR-symmetric when
\begin{align}\label{eq:TRS_Vform}
    V=\begin{bmatrix}
    V_{1} & -\bar{V}_{2} \\ V_{2} & \bar{V}_{1}
\end{bmatrix},
\end{align}
where $V_{1}$ is a $2M\times2 M$ Hermitian operator and $V_{2}$ is a $2M\times 2M$  anti-symmetric operator.

\paragraph{Case $M=1$ and $p=2$.}
Assume further that $E=3$ for the unperturbed system $H$, corresponding to four left propagating modes with indices $\{(1,2,+1),(2,0,+1),(2,1,+1),(2,2,+1)\}$ and four right propagating modes given by $\{(1,0,-1),(1,1,-1),(1,2,-1),(2,2,-1)\}$. Hence, for such an unperturbed system $H$,  we have $\Tr  (T_+^*T_+)=3=\Tr  (T_-^*T_-)$.

Since $p=2$ is even, we obtain $\ind_2\, H = \ind_2\, H_V=0$ from our theoretical section. We consider two perturbations respecting the FTR-symmetry and compute the one-sided transmissions against the length of the support of $V$. The numerical discretization is the same as one stated in Sec.\eqref{sec:4x4TRS}. The perturbation is constructed in the form of \eqref{eq:TRS_Vform}. In the first case, $V_1$ in \eqref{eq:TRS_Vform} is a $4\times 4$ zero matrix, while, 
\begin{align}\label{eq:eg3_V2}
    V_2(x,y;l)=(1+i)v_2(x,y)\sigma_2\chi_{[0,l]}=v_2(x,y)\begin{bmatrix}
    0 & 1-i \\ i-1 & 0 
\end{bmatrix}\chi_{[0,l]}
\end{align} where,
\begin{align}\label{v_p2}
    v_2(x,y)=(1+y)\exp(-y^2)\Big(&\cos((-E-\sqrt{E^2-8})x)+\cos((-E+\sqrt{E^2-8})x)\\ \nonumber
    &+\cos(2\sqrt{E^2-8}x)+\cos(2Ex)\Big).
\end{align}
From Fig.\ref{fig:eg3_M2}(a), we observe that the one-sided transmissions do not vanish and converge to the integer $2$.  However, if we set in  \eqref{eq:TRS_Vform}
 \begin{align}\label{eg3_p2_V1}
     V_1(x,y;l)=v_2(x,y)\sigma_3\chi_{[0,l]}=v_2(x,y)\begin{bmatrix}
    1 & 0 \\ 0 & -1 
\end{bmatrix}\chi_{[0,l]},
\end{align}
and keep $V_2$ as in \eqref{eq:eg3_V2}, then we observe in Fig.\ref{fig:eg3_M2}(b) a complete asymptotic Anderson localization with one-sided transmissions rapidly decaying to $0$. 

The latter example is consistent with our theoretical result that $\ind_2 H = \ind_2 H_V=0$. The former example is not. It shows that additional hidden symmetries may be present in the system forcing $\Tr  (T_+^*T_+)$ to transition from a value of $3$ in the absence of perturbations to an asymptotic value of $2$ in the presence of large perturbations that much reflect an additional hidden symmetry besides FRT. Perhaps surprisingly, for all other perturbations $V$ tested for $M=1$, $p=2$, the limit of the one-sided transmissions was always an even integer ($0$ or $2$). 

We now consider the non-FTR-symmetrix system $H_V=h + v_2\sigma_2\chi_{[0,l]}$. The perturbation is the same as in \eqref{eg3_p2_V1}. The current observable $2
\pi \sigma_I$ is equal to $-2$, which is no longer trivial as for FTR-symmetric operators. Fig.\ref{fig:eg3_M2}(c) presents one-sided transmissions. We observe that the limits of the left and right transmissions are $-2$ and $0$, respectively. This shows that an asymmetric transport with $2$ unit of quantized left travelling current is topologically protected from Anderson localization.
\begin{figure}[ht!]
    \centering
      \begin{subfigure}{0.32\textwidth}
		\includegraphics[width = \linewidth]{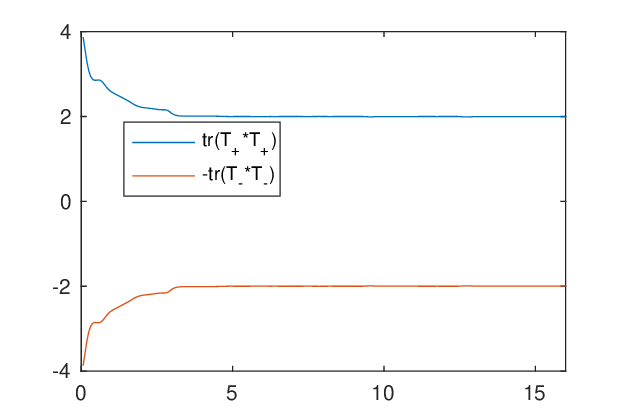}
		\caption{$h+\bar h$ model; $V_1=0,\,V_2$ in \eqref{eq:eg3_V2}}
  \end{subfigure}
\begin{subfigure}{0.32\textwidth}
		\includegraphics[width = \linewidth]{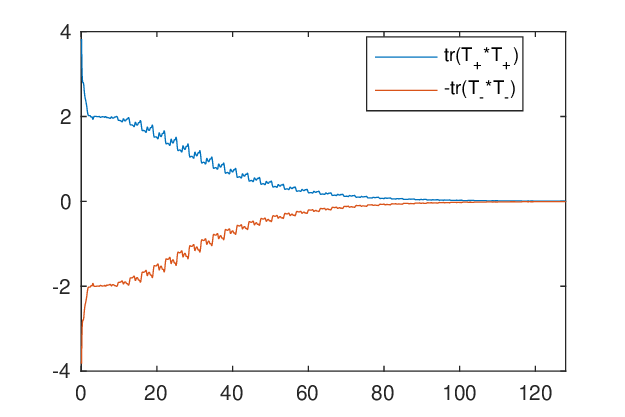}
		\caption{$h+\bar h$ model, $V_1$ specified in \eqref{eg3_p2_V1}}
  \end{subfigure}
     \begin{subfigure}{0.32\textwidth}
		\includegraphics[width = \linewidth]{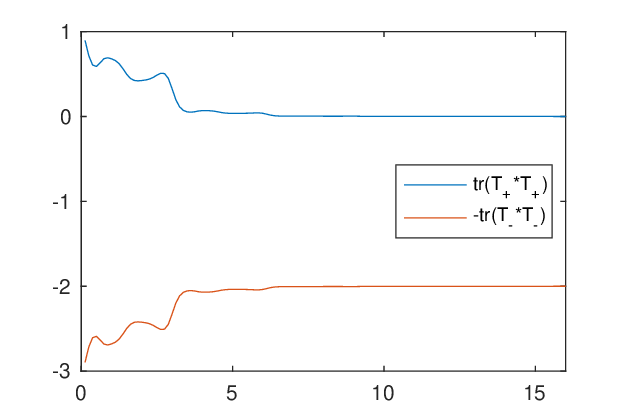}
		\caption{$h$ model, $V=v_2\sigma_3\chi_{[0,l]}$ }
  \end{subfigure}
\caption{One-sided transmissions  $\Tr(T_+^*T_+)$ and $-\Tr(T_-^*T_-)$ against the length of perturbation in case of $M=1$, $p=2$. }
    \label{fig:eg3_M2}
\end{figure}

\paragraph{Case $M=3$ and  $p=1$.}
As a final example, we consider the unperturbed system $H=(h\oplus\bar h)^{\oplus3}$ acting on a $12\times 12$ spinor. At energy $E=1.8$, the system has $9$ left (right) propagating modes. We consider two types of perturbation $V$ preserving the TRS and are constructed in the form \eqref{eq:TRS_Vform}. In both cases, since $ V_{2}$ defines the correlation between different $h$, we generally set
\begin{align}
     V_{2}(x,y;l)=v_1(x,y)\mathbf{1}_3\otimes \sigma_2\chi_{[0,l]}
\end{align}
where $\mathbf{1}_3$ is an $3\times 3$ matrix whose elements are all $1$. For $V_1$ in \eqref{eq:TRS_Vform},  in the first example, we assume,
\begin{align}\label{eg4_V1_EXS}
     V_{1}(x,y;l)=v_1(x,y) I_3\otimes \sigma_0\chi_{[0,l]}.
\end{align}
Fig.\ref{fig:eg4_M3}(a) presents the one-sided transmissions against the length of the support of the perturbation $l$. The limit of large randomness turns out to be equal to $3$. 

We now consider the same configuration with $V_1$ of the form:
 \begin{align}\label{eg4_V1_NEX}
     V_{1}=v_1 \begin{bmatrix}
         1 & &\\ & -1 & \\ & & 2
     \end{bmatrix}\otimes \sigma_0.
\end{align}
We observe in Fig.\ref{fig:eg4_M3}(b) that $\Tr(T_+^*T_+)$ now converges to the expected value $1$ instead of the larger value $3$ obtained in Fig.\ref{fig:eg4_M3}(a). When $V_1$ is constructed using \eqref{eg4_V1_EXS}, it turns out that the perturbed system $H_V$ admits three building blocks of $(h\oplus \bar h)$ that are exchangeable. Such exchange symmetry apparently prevents $\Tr(T_+^*T_+)$ from converging to an integer that is less than the number of building block $3$. This reflects another (not so) hidden symmetry of the system beyond the FTR symmetry.
\begin{figure}[ht!]
    \centering
      \begin{subfigure}{0.45\textwidth}
		\includegraphics[width = \linewidth]{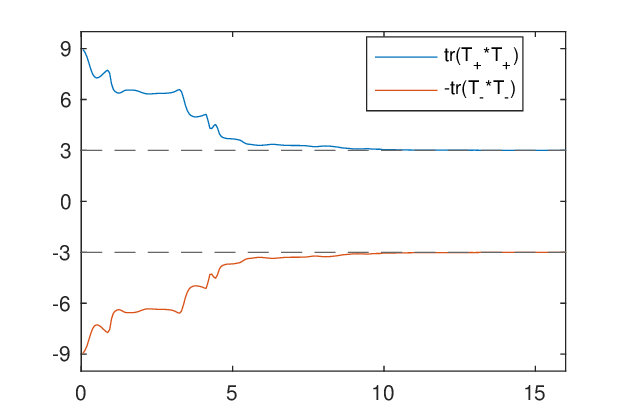}
		\caption{$V_1$ specified in \eqref{eg4_V1_EXS}}
  \end{subfigure}
\begin{subfigure}{0.45\textwidth}
		\includegraphics[width = \linewidth]{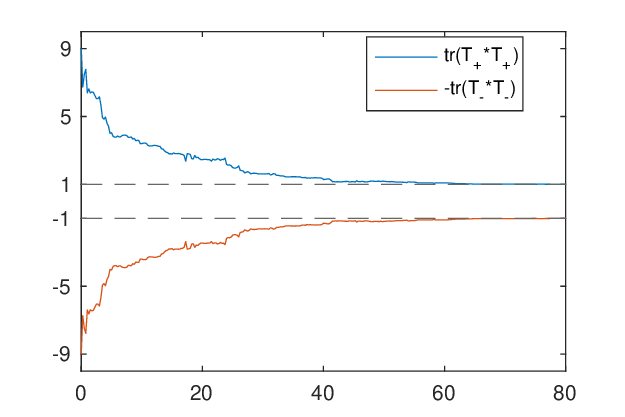}
		\caption{$V_1$ specified in \eqref{eg4_V1_NEX}}
  \end{subfigure}
\caption{One-sided transmissions  $\Tr(T_+^*T_+)$ and $-\Tr(T_-^*T_-)$ against the length of perturbation in case of $M=3$, $p=1$.  }
    \label{fig:eg4_M3}
\end{figure}

\section*{Acknowledgment} This work was supported in part by the US National Science Foundation Grants DMS-2306411 and DMS-1908736.

\bibliographystyle{siam}
\bibliography{ref.bib}

\end{document}

%% file: self_defined_command_gb.tex
\def \Rm {\mathbb R}
\def \Nm {\mathbb N}
\def \Cm {\mathbb C}
\def \Zm {\mathbb Z}

\def \Mm {\mathbb M}
\newcommand{\eps}{\varepsilon}

\newcommand{\dsum}{\displaystyle\sum}
\newcommand{\dint}{\displaystyle\int}

\newcommand{\mF}{\mathcal F}

\newcommand{\mH}{\mathcal H}

\newcommand{\mK}{\mathcal K}

\newcommand{\rI}{{\rm I}}
\newcommand{\rJ}{{\rm J}}

\newcommand{\fa}{{\mathfrak a}}

\newcommand{\fS}{{\mathfrak S}}

\newcommand{\VP}{Q}

\newcommand{\ind}{{\rm Index\,}}

\newcommand{\cout}[1]{}

\newcommand{\sgn}[1]{\,{\rm sign}(#1)}

\newcommand{\Tr}{{\rm Tr}}

\newcommand{\rP}{{\rm P}}

\newcommand{\epm}{\epsilon_m}
\newcommand{\psin}{\psi_{\rm in}}
\newcommand{\psiout}{\psi_{\rm out}}

\newcommand{\R}{{\rm R}}

\newcommand{\trr}{{\rm tr}}

\makeatletter
\let\save@mathaccent\mathaccent
\newcommand*\if@single[3]{%
  \setbox0\hbox{${\mathaccent"0362{#1}}^H$}%
  \setbox2\hbox{${\mathaccent"0362{\kern0pt#1}}^H$}%
  \ifdim\ht0=\ht2 #3\else #2\fi
  }
\newcommand*\rel@kern[1]{\kern#1\dimexpr\macc@kerna}
\newcommand*\widebar[1]{\@ifnextchar^{{\wide@bar{#1}{0}}}{\wide@bar{#1}{1}}}
\newcommand*\wide@bar[2]{\if@single{#1}{\wide@bar@{#1}{#2}{1}}{\wide@bar@{#1}{#2}{2}}}
\newcommand*\wide@bar@[3]{%
  \begingroup
  \def\mathaccent##1##2{%
    \let\mathaccent\save@mathaccent
    \if#32 \let\macc@nucleus\first@char \fi
    \setbox\z@\hbox{$\macc@style{\macc@nucleus}_{}$}%
    \setbox\tw@\hbox{$\macc@style{\macc@nucleus}{}_{}$}%
    \dimen@\wd\tw@
    \advance\dimen@-\wd\z@
    \divide\dimen@ 3
    \@tempdima\wd\tw@
    \advance\@tempdima-\scriptspace
    \divide\@tempdima 10
    \advance\dimen@-\@tempdima
    \ifdim\dimen@>\z@ \dimen@0pt\fi
    \rel@kern{0.6}\kern-\dimen@
    \if#31
      \overline{\rel@kern{-0.6}\kern\dimen@\macc@nucleus\rel@kern{0.4}\kern\dimen@}%
      \advance\dimen@0.4\dimexpr\macc@kerna
      \let\final@kern#2%
      \ifdim\dimen@<\z@ \let\final@kern1\fi
      \if\final@kern1 \kern-\dimen@\fi
    \else
      \overline{\rel@kern{-0.6}\kern\dimen@#1}%
    \fi
  }%
  \macc@depth\@ne
  \let\math@bgroup\@empty \let\math@egroup\macc@set@skewchar
  \mathsurround\z@ \frozen@everymath{\mathgroup\macc@group\relax}%
  \macc@set@skewchar\relax
  \let\mathaccentV\macc@nested@a
  \if#31
    \macc@nested@a\relax111{#1}%
  \else
    \def\gobble@till@marker##1\endmarker{}%
    \futurelet\first@char\gobble@till@marker#1\endmarker
    \ifcat\noexpand\first@char A\else
      \def\first@char{}%
    \fi
    \macc@nested@a\relax111{\first@char}%
  \fi
  \endgroup
}
\makeatother